\documentclass[journal]{IEEEtran}

\usepackage[T1]{fontenc}
\makeatletter
\let\l@ENGLISH\l@english
\makeatother
\usepackage[english]{babel}
\usepackage[babel]{csquotes}
\usepackage{lipsum}
\usepackage{url}

\usepackage{emptypage}
\usepackage[english]{varioref}
\usepackage{color}
\usepackage{eurosym}
\usepackage{lettrine}

\usepackage{amsmath}
\usepackage{amsthm}
 \usepackage{multirow}
\usepackage{footmisc} %increases the number of footnotes from 10 to 15

\DeclareMathOperator{\diag}{diag}
\usepackage{mathrsfs}
\usepackage{amsfonts}
\usepackage{amssymb}
\newtheorem{theorem}{Theorem}
\newtheorem{proposition}{Proposition}
\newtheorem{lemma}{Lemma}

\newtheorem{corollary}{Corollary}
\newtheorem{obj}{Objective}
\newtheorem{remark}{Remark}
\newtheorem{assumption}{Assumption}
\usepackage{algorithm}
\usepackage{dsfont}
\usepackage{algpseudocode}
\usepackage{upgreek}
\newcommand{\numberset}{\mathbb}

\newcommand{\R}{\numberset{R}}

\usepackage{mathtools}

\usepackage{amscd}
\usepackage{booktabs}
\usepackage{graphics}
\usepackage{epsfig}
\usepackage{contour}
\usepackage{bm}
\usepackage{rotating}
\usepackage{tabularx}
\usepackage{longtable}
\usepackage{pdflscape}
\usepackage{cite}
\usepackage{listings}
\usepackage{xcolor}
\usepackage{siunitx}
\definecolor{green}{rgb}{0.13,0.55,0.13}
\renewcommand{\vec}{\bm}

\newcommand{\red}[1]{{\textcolor{black}{#1}}}
\newcommand{\green}[1]{{\textcolor{black}{#1}}}
\newcommand{\blue}[1]{{\textcolor{black}{#1}}}
\newcommand{\magenta}[1]{{\textcolor{black}{#1}}}

\DeclareMathAlphabet{\mathpzc}{OT1}{pzc}{m}{it}
\DeclareMathAlphabet{\mathcal}{OMS}{cmsy}{m}{n}

\newcommand{\V}{\mathcal{V}}

\usepackage{cases}
\usepackage{tikz}
\usetikzlibrary{arrows,automata}
\usetikzlibrary{arrows,shapes,backgrounds,calc,positioning,patterns}
\usepackage{balance}
\usetikzlibrary{calc,arrows,shapes,backgrounds,calc,positioning,patterns,decorations.pathmorphing,decorations.markings,mindmap,trees}
\tikzstyle{block} = [draw, rectangle, minimum height=2em, minimum
width=4em] \tikzstyle{sum} = [draw, fill=blue!20, circle, node
distance=1cm] \tikzstyle{input} = [coordinate] \tikzstyle{output} =
[coordinate] \tikzstyle{pinstyle} = [pin edge={to-,thin,black}]
\usepackage{blox}
\usepackage{cases}
\usepackage{framed}
\colorlet{shadecolor}{black!15}
\usepackage{bigints}
\usepackage[american,cute inductors,smartlabels]{circuitikz}
\usetikzlibrary{arrows,shapes,backgrounds,calc,positioning,patterns}

\usepackage[american,cuteinductors,smartlabels]{circuitikz}

\usetikzlibrary{calc}
\ctikzset{bipoles/thickness=1} \ctikzset{bipoles/length=0.8cm}
\ctikzset{bipoles/diode/height=.375}
\ctikzset{bipoles/diode/width=.3}
\ctikzset{tripoles/thyristor/height=.8}
\ctikzset{tripoles/thyristor/width=1}
\ctikzset{bipoles/vsourceam/height/.initial=.7}
\ctikzset{bipoles/vsourceam/width/.initial=.7} \tikzstyle{every
node}=[font=\small] \tikzstyle{every path}=[line width=0.8pt,line
cap=round,line join=round]
\newcommand{\ncom}{\newcommand}
\ncom{\beqn}{\begin{eqnarray*}} \ncom{\eeqn}{\end{eqnarray*}}
\ncom{\beq}{\begin{eqnarray}} \ncom{\eeq}{\end{eqnarray}}

\newcommand{\Ra}[1]{{\textcolor{black}{#1}}}%blue
\newcommand{\Rc}[1]{{\textcolor{black}{#1}}}%magenta
\newcommand{\Reva}[1]{{\textcolor{black}{#1}}}%blue
\newcommand{\Revb}[1]{{\textcolor{black}{#1}}}%red
\newcommand{\redd}[1]{{\textcolor{black}{#1}}}%green
\usepackage{flushend}
\begin{document}
    \title{Differentiation and Passivity for Control of Brayton-Moser Systems}
    \author{Krishna Chaitanya Kosaraju,
        Michele~Cucuzzella,
        Jacquelien M. A. Scherpen
        and Ramkrishna Pasumarthy
        \thanks{K. C. kosaraju is with Department of Electrical Engineering, University of Notre Dame, 275 Fitzpatrick Hl Engrng, Notre Dame, IN 46556, USA.}
        \thanks{M. Cucuzzella and J. M. A. Scherpen are with the Jan C. Wilems Center for Systems and Control, ENTEG, Faculty of Science and Engineering, University of Groningen, Nijenborgh 4, 9747 AG Groningen, the Netherlands.} 
        \thanks{R. Pasumarthy is with Department of electrical engineering, Indian Institute of Technology-Madras, Chennai-36, India. (email: \tt{k.c.kosaraju@nd.edu}, \tt{\{m.cucuzzella, j.m.a.scherpen\}@rug.nl}, \tt{ramkrishna@ee.iitm.ac.in})}
        \thanks{
        This work is  supported by the Netherlands Organisation for Scientific Research through
        Research Programme ENBARK+ (under Project 408.urs+.16.005) and the EU Project `MatchIT' (project number: 82203).}}
    \maketitle
    %
    %
    %
    %
    %========================================ABSTRACT
    \begin{abstract}
This paper deals with a class of Resistive-Inductive-Capacitive (RLC) circuits and switched RLC (s--RLC)  circuits modeled in the Brayton Moser framework. 
For this class of systems, new passivity properties using a Krasovskii-type Lyapunov function as storage function are presented, where the supply-rate is  function of the system states, inputs and their first time derivatives.
Moreover, after showing the integrability property of the port-variables, two \emph{simple} control methodologies called {\em output shaping} and {\em input shaping} are proposed for regulating the voltage in RLC and s--RLC circuits. 
Global asymptotic stability is theoretically proved for both the proposed control methodologies. Moreover, robustness with respect to load uncertainty is ensured by the {\em input shaping} methodology.
The applicability of the proposed methodologies is illustrated by designing voltage controllers for DC-DC converters and DC networks.
    \end{abstract}
    \IEEEpeerreviewmaketitle
    %
    %
    %
    %
    %========================================INTRODUCTION
    \section{Introduction}

In the recent years, passivity theory has gained renewed attention because of its advantages and practicality in modeling and control of multi-domain dynamical systems \cite{5184955,915398}. In general, a system is passive if  there exists a (bounded from below) {\em storage function} $S:\mathbb{R}^n\rightarrow \mathbb{R}_{+}$ satisfying 
\begin{equation}\label{passivity_gen}
S(x(t))-S(x(0))\leq \int_{0}^{t}u^\top y dt,
\end{equation}
where $x\in \mathbb{R}^n$ is the system state, $u, y\in \mathbb{R}^m$ are the input and the output, also called {\em port-variables} and the product $u^\top y$ is commonly known as {\em supply-rate} \cite{van2000l2, HILL1980327}. Naturally, one can interpret the storage function as the total system energy, and the supply rate as the power supplied to the system.  
Consequently, inequality \eqref{passivity_gen} implies that the newly \emph{stored} energy is never greater than the \emph{supplied} one. 

In order to analyze the passivity properties of a general nonlinear system, it is usually required to be artful in designing the storage function. For this reason, it is helpful to recast the system dynamics into a known framework, such as the port-Hamiltonian (pH) one \cite{SYS-002}, where the storage function, also called Hamiltonian function, generally depends on the system energy.
Another well known framework that has been extensively explored for modeling of nonlinear Resistive-Inductive-Capacitive (RLC) circuits, is the Brayton-Moser (BM) framework \cite{BMnetworks1,BMnetworks2}, where the storage function relies on the system power (see \cite{1393170} for further details on geometric modeling of nonlinear RLC circuits).  
 
Nowadays, power-converters play a prominent role in smart grids. Conventional power-converters consist of (passive) subsystems interconnected through switches. 
%Switching devices are complex hybrid systems that need to be controlled in order to regulate the energy transfer.  
%Under the condition
%that the Pulse Width Modulation (PWM) frequency is sufficiently
%high, the state of the system can be replaced by the corresponding average state
%representing the average inductor currents and capacitor voltages,
%while the switching control input is replaced by the so called duty
%cycle of the converter {\cite{ESCOBAR1999445}}.
In this paper, we consider a large class of \emph{ switched} RLC (s--RLC) circuits, which models the majority of the existing power converters (e.g., buck, boost, buck-boost and C\'uk).
%  \subsection{Literature Review}
 Although the analysis of s--RLC circuits has received a significant amount of attention (see for instance \cite{4154602,ESCOBAR1999445,577568a, SCHERPEN2003365} and the references therein), we notice that results based on the \emph{passivity properties} of the open-loop system are still lacking.  
On the other hand, a significant number of results have been published relying on Passivity-Based Control (PBC)~\cite{ORTEGA2004432,GARCIACANSECO2010127,Ortega1998,MEHRA201798,7993047,BORJA2016230,chinde2017passivity}, where the main idea is generally to passify the controlled system such that the  closed-loop storage function has a minimum at the desired operating point~\cite{915398}. 
%Although this objective can be achieved in different ways, the most used methodology is the Interconnection and Damping Assignment (IDA)-PBC technique \cite{ORTEGA2004432}, and the main goal is to assign a desired structure to the closed-loop system \cite{GARCIACANSECO2010127,Ortega1998}. Another recent technique, relevant to this work, is called Proportional-Integral-Derivative (PID)-PBC \cite{ MEHRA201798}, and the main goal is to find a (cyclo-) ~passive map with a suitable output, such that the (integrated) output port-variable can be used to shape the closed-loop storage function  \cite{7993047,BORJA2016230,chinde2017passivity}.
%The existing PBC techniques can be broadly classified into two categories. The first one is called Interconnection and Damping Assignment (IDA)-PBC technique \cite{ORTEGA2004432}, and the main goal is to assign a desired structure to the closed-loop system \cite{ORTEGA2002585,GARCIACANSECO2010127,Ortega1998}. 
%The second one is called Proportional-Integral-Derivative (PID)-PBC \cite{DONetal, MEHRA201798}, and the main goal is to find a (cyclo-)passive map with a suitable output, such that the (integrated) output port-variable can be used to shape the closed-loop storage function  \cite{7993047,BORJA2016230,chinde2017passivity, 7846443}. 
%Existing results on passivity based control of s--RLC circuits are predominantly based on IDA-PBC technique \cite{ Ortega1998, 912349, 954388}. 
% However, these results are only local. 
 \green{%A key reason is that, the well-known and well-developed port-Hamiltonian framework not-only results in an complicated representation\cite{ESCOBAR1999445} (also see Remarks \ref{rem:ph_srlc} of the manuscript) but-also uses flux and charge as state-variable which are not measurable. As a consequence, 
However, the passivity properties and the control techniques developed for pH systems cannot   be straightforwardly applied  to s--RLC networks.}
%(Remark \ref{energy_s-rlc} shows that passivity properties derived by using the system energy as storage function are limited to a small class of systems).}
 \green{Alternatively, in \cite{1323174} the authors formalize the use of the BM framework for  analyzing s--RLC circuits and also provide tuning rules based on the well-known BM theorems. }%\cite{jeltsema2005modeling}
% In \cite{1323174}, the authors also proposed tuning rules using the Brayton-Moser stability theorems \cite{BMnetworks,jeltsema2005modeling}. 

\subsection{Motivation and Main Contributions}
Lyapunov theory is  fundamental  in systems theory. In order to study the stability of a dynamical system, one generally needs to find a suitable Lyapunov function. 
%However, finding a Lyapunov function mostly comes from domain knowledge and some times purely art. 
Krasovskii proposed a simple and elegant candidate Lyapunov function, {where one needs to compute some point-wise conditions for sufficiency of Lyapunov stability \cite{khalil2002nonlinear}}. In a similar manner, passivity theory hinges on finding candidate storage functions satisfying \eqref{passivity_gen}. However, the candidate Lyapunov function proposed by Krasovskii is not well explored as a storage function. 
In \cite{mtns, 7846443,8431720}, the authors presented a preliminary result on the passivity property for a class of RLC circuits using such a storage function, named {\em Krasovskii storage function}, i.e.,
\begin{equation}
\label{Krasovskii}
S(x,\dot{x})=\dfrac{1}{2}\dot{x}^\top M(x) \dot{x},
\end{equation}
where $M(x)>0 \in \mathbb{R}^{n\times n}$. 
\blue{By using such a storage function, in this paper we present completely new passivity properties for a class of s--RLC circuits and for a class of RLC circuits wider than the one  considered in \cite{mtns, 7846443, 8431720}\footnote{\Ra{Note that in \cite{mtns}, the authors have \emph{only} explored as a conclusive remark (see \cite[Section V]{mtns}) the idea of using \eqref{Krasovskii} as storage function for a particular electrical example. Moreover,  in \cite{7846443} and \cite{8431720}, \emph{only} a preliminary result on the passivity property of \emph{only} RLC circuits is established under some assumptions that are more restrictive than the ones in this paper.}}.}

The output port-variables associated with the above storage function have integrability properties. It is well-established that the integrated output port-variable can be used to shape the closed-loop storage function. This leads to the development of a new control technique, named {\em output shaping}.
More precisely, the above storage function allows to establish a passivity property with supply-rate depending on the system state $x$, input $u$ and also their first time derivatives $\dot{x}, \dot{u}$. 
This enables us to develop a second control technique that we call \emph{input shaping}, which is radically new in PBC methodology. More precisely, we use the integrated input port-variable to shape the closed-loop storage function.
%In this, we use the integrated input port-variable to shape the closed-loop storage function.
%Moreover, the authors  also infer that the new passive map has a similar structure to the differential-passivity\cite{6760930}. This can be followed from the fact that differential-passivity stems from contraction theory and further contraction-theory originates from Krasovskii type Lyapunov function. Moreover, the storage function is in terms of states and velocities, hence, the new passivity property considers the {\em extended dynamics} (but extended with respect to time).
%In the case of RLC circuits, the passivity is now with respect to the time derivative of controlled voltages or controlled currents and time derivative of currents or voltages respectively. This property has three particular advantages when compared to port-Hamiltonian and Brayton-Moser based passive maps:
Furthermore, a Krasovskii storage function has the following advantages:
\begin{itemize}
\item[(i)] Since the supply-rate is a function of the first time derivative of the system state and input, the so-called \emph{dissipation obstacle}\footnote{For a system with non-zero supply-rate at the desired operating point, the controller has to provide unbounded energy to stabilize the system. In the literature, this is usually referred to as {\em dissipation obstacle} or {\em pervasive dissipation}.} problem \cite{915398} is avoided.
\item[(ii)] There are no parametric constraints that usually appear in Brayton-Moser framework (see for instance {\cite[Theorem 1]{1230225}}).
\item[(iii)] The port-variables are integrable.%, which helps us avoid the complex partial differential equation in IDA-PBC technique.
\end{itemize}
\smallskip 
      %In this work, we consider a class of RLC and s--RLC circuits. 
      Below, we list the main contributions of this work:
\begin{itemize}
	\item[(i)] The use of a storage function similar to \eqref{Krasovskii} for s--RLC circuits leads to a new passive map useful for control purposes. 
	%The resulting port-variables have the following characteristics: the input port-variable is a function of the time derivative of the duty ratio, and the output port-variable has units of power/time and is a function of states and their time derivatives.
	\item [(ii)] We use the integrated port-variables to shape the closed-loop storage function and propose two \emph{simple} control techniques: \emph{output shaping} and \emph{input shaping}. 
	%More precisely, the output shaping utilizes the integrability property of the output and vice versa. 
	Both the techniques are used for regulating the voltage in RLC and s--RLC circuits.
	\item [(iii)] The input shaping technique is robust with respect to load uncertainty and requires less assumptions on the system parameters/structure than the output shaping one. 
	\end{itemize}
	%Moreover, we extend the aforementioned technique in order to include a constant current load.
	\smallskip
The proposed techniques are finally illustrated with application to buck, boost, buck-boost, C\'uk  DC-DC converters and DC networks, which are attracting growing interest and receiving much research attention \cite{Zhao2015,J_Cucuzzella_18,DePersis2016,8071020,J_Trip2018,cucuzzella2017robust}. Simulation results show excellent performance.

\Reva{Note that the BM framework adopted in this work to model RLC and s--RLC circuits represents a larger class of nonlinear gradient systems \cite{BROER20101,VANDERSCHAFT20113321}. For instance, in {\cite{5184955}},  the BM equations are shown to be applicable to a wide class of nonlinear physical systems, including lumped-parameter mechanical, fluid, thermal, and electromechanical systems, electrical power converters, mechanical systems with impacts and distributed-parameter systems \cite{chaitanya2019modeling} {(see also {\cite{jeltsema2002brayton,4154602,favache2010power, 7155514}} for further applications)}.
%Moreover, for practical reasons, it is desired to describe the dynamics of electrical circuits by using co-energy variables (i.e., current and voltage) instead of energy variables (i.e., flux and charge). 
	Moreover, a practical advantage of using the BM framework is that the system variables are directly expressed in terms of easily measurable physical quantities, such as currents, voltages, velocities, forces, volume flows, pressures, or temperatures. On the other hand, the Lagrangian and Hamiltonian formulations normally involve generalized displacement and momenta, which in many cases cannot be measured directly. Furthermore, s--RLC circuits do not generally inherit a standard pH structure because the interconnection matrix is often a function of both the system state and control input, rather than only the state~\cite{ESCOBAR1999445}. As a consequence, the existing passivity-based techniques may not be useful for the analysis and control purpose.}

    \subsection{Outline}
This paper is outlined as follows. In Section II, we recall the BM representation of RLC and s--RLC circuits and formulate the control objective after introducing the required assumptions. In Section III, we present the newly established passivity property for the RLC and s--RLC circuits. Then, using these properties, we propose two novel control techniques: output shaping and input shaping. In Section IV and Appendix, we illustrate the proposed techniques on buck, boost,  buck-boost, C\'uk DC-DC converters and DC networks with buck and boost converters interconnected through resistive-inductive lines. Finally, we conclude and present some possible future directions in Section V.
    \subsection{{Notation}}
    The set of real numbers is denoted by $\R$. The set of positive real numbers is denoted by $\R_{+}$.
    Let $x\in \mathbb{R}^n$ and $y\in \mathbb{R}^m$. Given a mapping $f:\mathbb{R}^n\times \mathbb{R}^m\rightarrow \mathbb{R}$, the symbol $\nabla_xf(x,y)$ and $\nabla_yf(x,y)$ denotes the partial derivative of $f(x,y)$ with respect to $x$ and $y$ respectively. 
    %Moreover, given the signal $z_s \in \R^m$, the subscript $s$ means that $z$ is a source signal. 
    Let $K\in \mathbb{R}^{n\times n}$, then $K>0$ and $K\geq 0$ denote that $K$ is symmetric positive definite and symmetric positive semi-definite, respectively. Assume $K>0$, then $||x||_{K}:=\sqrt{x^\top K x}$ and $||K||_s$ denotes the spectral norm of $K$. Let $Q_1$ and $Q_2$ denote square matrices of order $m$ and $n$ respectively. Then $\diag\{Q_1,Q_2\}$ denotes a block-diagonal matrix of order $m+n$ with block entries $Q_1$ and $Q_2$. Given ${p} \in \R^n$ and ${q} \in \R^n$, \lq$\circ$' denotes the so-called
Hadamard product (also known as Schur product), i.e., $({p} \circ {q}) \in \R^n$ with $({p}\circ {q})_i := p_i q_i, i=1,\dots, n$. Moreover, $[p] := \diag\{p_1, \dots, p_n\}$.
    %
    %
    %
    %
    %========================================PRELIMINARIES
  %  \section{Preliminaries on Brayton-Moser Equations}
        \section{Preliminaries and Problem Formulation}
    \label{sec:preliminaries}
%    Consider the following nonlinear system in state-space representation
%    \begin{eqnarray} \label{gen_non_sys}
%    	\begin{matrix}
%    	\dot{x}&=&f(x)+g(x)u\\
%    	y&=&h(x)
%    	\end{matrix}
%    \end{eqnarray}

    In this section, we briefly outline the Brayton-Moser (BM) formulation of RLC circuits and extend it to the case including an ideal switching element. 
    %Moreover, we will remark some limitations of existing methodologies.
    %
    %In the context of electrical networks, the passivity property is now established with respect to voltage and derivative of current, or current and derivative of voltage.  The following Proposition summarizes this passivity property for a class of RLC networks \eqref{BM_non_switched} modeled in Brayton-Moser framework.
    %
    %Where $x=(I,V)$, $Q=diag\{-L,C\}$ and $\mathcal{B}=[-B^\top\;O ]^\top$, with $O\in\mathcal{R}^{m\times \rho}$ represents a zero-matrix. 
    %
     %   In this section, we will briefly outline the Brayton-Moser formulation of RGLC circuits and extend it to include an ideal switching element. In the process, we will remark some limitations in existing methodologies. %passivity based control methods.%methods for passivity and control. \\
    \subsection{Non-Switched Electrical Circuits}
    Consider the class of {topologically complete} RLC circuits~\cite{1323174} with $\sigma$ inductors, $\rho$ capacitors {and $m$ (current-controlled) voltage sources  $u_s \in \mathbb{R}^m$ connected in series with inductors}. In \cite{BMnetworks1,BMnetworks2}, Brayton and Moser show that the dynamics\footnote{\Revb{For further details and a large number of
		examples, we suggest the reading of the sidebar \lq History of the Mixed-Potential Function' and
		section \lq The Brayton-Moser equations' in \cite{5184955}.}} of this class of systems can be represented as
    %\footnote{To present our ideas we have only considered In equation \eqref{BM_non_switched}, we have considered only voltage sources. However, we would like to note that  having a current souconsidered current sources connected in parallel to the capacitors. }
    \begin{align}
    \label{BM_non_switched}
    \begin{split}
        -L\dot{I}&=\nabla_IP(I,V) - Bu_s\\
        C\dot{V}&=\nabla_VP(I,V),
    \end{split}
    %\label{BM_non_switched}
    \end{align}
where $L\in \mathbb{R}^{\sigma \times \sigma}$ and $C\in
\mathbb{R}^{\rho \times \rho}$ are the inductances and capacitances matrices,
respectively. The state variables $I\in \mathbb{R}^{\sigma}$ and $V
\in \mathbb{R}^{\rho}$ denote the currents through the $\sigma$
inductors and the voltages across the $\rho$ capacitors,
respectively. The matrix $B\in \mathbb{R}^{\sigma\times m}$ is the
input matrix with full column rank and $P:\mathbb{R}^{\sigma\times \rho}\rightarrow
\mathbb{R}$ represents the so-called \emph{mixed-potential}
function, given by,
    \begin{equation}
    \label{Mixpot}
    P(I,V) = I^\top\Gamma V+P_R(I)-P_G(V),
    \end{equation}
    where {$\Gamma\in \mathbb{R}^{\sigma\times \rho}$} captures the power circulating across the dynamic elements. The resistive content $P_R:\mathbb{R}^\sigma \rightarrow \mathbb{R}$ and the resistive co-content $P_G:\mathbb{R}^\rho \rightarrow \mathbb{R}$ capture the power dissipated in the resistors  connected in series to the inductors and in parallel to the capacitors, respectively. 
    \begin{remark}[Current sources]
    	For the sake of simplicity, in \eqref{BM_non_switched} we have not included current sources. However, the results presented in this note can also be developed for current sources in a straightforward manner.
    \end{remark}
    According to the BM formulation, system~\eqref{BM_non_switched} can compactly be written as
      \begin{equation}\label{gradient_structure}
Q\dot{x}= \nabla_xP(x)+\tilde{B}u_s,
    \end{equation}
 where $x=(I^\top,\;V^\top)^\top$, $Q=\diag\{-L,C\}$ and $\tilde{B}=(-B^\top\;O)^\top$,  $O\in\mathbb{R}^{m\times \rho}$ being a zero-matrix. 
 To permit the controller design in the following sections, we introduce the following assumptions:
    \begin{assumption}[Inductance and capacitance matrices]
    \label{ass:LC}
        Matrices $L$ and $C$ are constant, symmetric\footnote{Matrices $L$ and $C$ can possibly capture mutual inductances and capacitances, respectively.} and positive-definite. 
    \end{assumption}
    \begin{assumption}[Resistive content and co-content]
    \label{ass:diss_potential}
    The resistive content and co-content of current controlled resistors $R$ and voltage controlled resistors $G$ are quadratic in $I$ and $V$ respectively, i.e., 
%        The dissipative-current potential $P_R(I)$ and the dissipative-voltage potential $P_G(V)$ are positive-definite and quadratic in $I$ and $V$ respectively, i.e.,
        \begin{equation}
        P_R(I)=\frac{1}{2}I^\top R I, \qquad P_G(V)=\frac{1}{2}V^\top G V,
        \end{equation}
        where $R \in \R^{\sigma \times \sigma}$ and $G \in \R^{\rho \times \rho}$ are positive semi-definite matrices.% are positive definite matrices.
    \end{assumption}
    %
    % \begin{itemize}
    % \item [(Assumption 1)] Matrices $L$, $C$ are  constant.
    % \item [(Assumption 2)] Mappings $R(I)$ and $G(V)$ are convex.
    % \end{itemize}
    % %
    \noindent Under Assumptions \ref{ass:LC} and \ref{ass:diss_potential}, it can be shown that system~\eqref{BM_non_switched} is passive with respect to the {\em power-conjugate}\footnote{We use the expression \emph{power-conjugate} to indicate that the product of input and output has units of power.} port-variables $u_s$, $B^\top I$ and the total energy stored in the network as storage function (see Remark \ref{energy_s-rlc}). 
    \subsection{(Average) Switched Electrical Circuits}
    We now consider the class of RLC circuits including an ideal switch\footnote{For the sake of simplicity we restrict the analysis to RLC circuits including only one switch. However, in Section \ref{sec:networks} we analyze a DC network including an arbitrary number of switches.} (s--RLC). {Let $u_d\in \{0,1\}$ and {$V_s \in \R^m$} denote the state of the switching element, i.e., \emph{open} or \emph{closed}, and the (current-controlled) voltage sources, respectively}. To describe the dynamics of s--RLC circuits we adopt the BM formulation \eqref{BM_non_switched} with the mixed-potential function and input matrix depending on the state of the switching element, i.e., $P:\{0,1\}\times \mathbb{R}^\sigma\times \mathbb{R}^\rho \rightarrow \mathbb{R}$ and $B:\{0,1\}\rightarrow \mathbb{R}^{\sigma\times m}$ can be expressed as
    \begin{align}
    \label{switched_sys_model}
        \begin{split}
            P(u_d,I,V)&=u_dP_1(I,V)+(1-u_d)P_0(I,V)\\
            B(u_d)&=u_dB_1+(1-u_d)B_0,
        \end{split}
    \end{align}
where $P_1(I,V)$, $B_1$ and $P_0(I,V)$, $B_0$ represent the
mixed-potential function and the input matrix of the s--RLC
circuit when $u_d=1$ and $u_d=0$,  respectively. Under the reasonable assumption
that the Pulse Width Modulation (PWM) frequency is sufficiently
high, the state of the system can be replaced by the corresponding average state
representing the average inductor currents and capacitor voltages,
while the switching control input is replaced by the so called duty
cycle of the converter \cite{ESCOBAR1999445}. For the sake of notational simplicity, from
now let $I$, $V$ and $u\in [0,1]$ denote the average
signals of $I$, $V$ and $u_d$, respectively, throughout the rest of the paper.
Consequently, the average behaviour of a  s--RLC electrical circuit can
be represented by the following BM equations
\begin{align}
    \label{BM_switched}
    \begin{split}
        -L\dot{I}&=\nabla_IP(u,I,V) - B(u)V_s\\
        C\dot{V}&=\nabla_VP(u,I,V).
    \end{split}
    %\label{BM_non_switched}
    \end{align}
    
    \begin{table}
                \centering
                \caption{Description of the used symbols}
                {\begin{tabular}{ll}
                    \toprule
                    &State variables\\
                    \midrule
                    $I$                     & Inductor current\\
                    $V$                     & Capacitor voltage\\
                    \midrule
                    & Parameters\\
                    \midrule
                    $L$                     &  Inductance\\
                    $C$                     & Capacitance\\
                    $G$                 & Conductance\\
                    $R$                 & Resistance \\
                    \midrule
                    & Inputs \\
                    \midrule
                    $u_s$					  & Control input (RLC circuits)\\
                    $u$                     & Duty cycle (s--RLC circuits) \\
                    $V_s$                    & Voltage source (s--RLC circuits)\\
                    \bottomrule
                    \end{tabular}}
                \label{tab:symbols}
                \end{table}
    
\begin{remark}[Resistive content and co-content structure]\label{rem::Dissipation_structure}
Note that if the content and co-content structure is not
affected by the switching signal, the mixed-potential function in \eqref{switched_sys_model} can
 be rewritten as follows
        \begin{equation}
        \label{switched_sys_modela}
                P(u,I,V)=I^\top\Gamma(u) V+P_R(I)-P_G(V),
        \end{equation}
        where the mapping $\Gamma:[0,1]\rightarrow \mathbb{R}^{\sigma\times \rho}$ is defined as
        \begin{equation}
        \label{eq:Gamma}
            \Gamma(u)=u\Gamma_1+(1-u)\Gamma_0,
        \end{equation}
        and $\Gamma_1, \Gamma_0$ capture the interconnection of the storage elements (i.e., inductors and capacitors) when $u=1$ and $u=0$, respectively.
    \end{remark}
    In the following we consider that the resistive content and co-content structure is not affected by the switching signal. As a consequence, system \eqref{BM_switched}
can be written as
\begin{align}
\label{BM_switched2}
\begin{split}
-L\dot{I} &=RI + \Gamma(u)V - B(u)V_s\\
C\dot V &=\Gamma^\top(u)I - GV.
\end{split}
\end{align}
The main symbols used in  \eqref{BM_non_switched}--\eqref{BM_switched2} are described
in Table~\ref{tab:symbols}. 

 \begin{remark}[Total energy as storage function]\label{energy_s-rlc}
	%Most physical system are passive with their total energy as storage function. For example, the RLC circuits presented in 
	It can be shown that the RLC circuit~\eqref{BM_non_switched} is passive with respect to the storage function
	\begin{equation}\label{total_energy}
	S_e(I,V)=\dfrac{1}{2}I^\top LI+\dfrac{1}{2}V^\top C V,
	\end{equation}
	and port-variables $u_s$ and $B^\top I$ (see for instance {\cite{915398}}). Consider now the s--RLC circuit \eqref{BM_switched2}. The first time derivative of the storage function \eqref{total_energy} along the solutions to \eqref{BM_switched2} satisfies 
	\begin{equation*}
	\dot{S}_e \leq
%	&= I^\top L\dot{I}+V^\top C\dot{V}\\
%	&= -I^\top\left(RI + \Gamma(u)V - B(u)V_s\right)+V^\top\left(\Gamma^\top(u)I - GV\right)\\
%	&= -I^\top RI-V^\top GV +I^\top B(u)V_s\\
%	&\leq  I^\top B(u)V_s\\
%	&=  
	I^\top B_0V_s+ uI^\top \left(B_1-B_0\right)V_s.
	\end{equation*}
	Consequently, system \eqref{BM_switched2} is passive with respect to the storage function \eqref{total_energy} and supply rate $uI^\top B_1V_s$ if and only if $B_0=0$ and $B_1\neq 0$. However, if we consider for instance the model of the boost converter (see Section \ref{subsec:boost}), the conditions $B_0=0$ and $B_1\neq 0$ are not satisfied. Furthermore, even supposing that the conditions $B_0=0$ and $B_1\neq 0$ hold, we notice that the supply rate $uI^\top B_1V_s$ is generally not equal to zero at the desired operating point, implying the occurrence of the so-called \lq dissipation obstacle' problem \cite{915398}.
\end{remark}

%In passivity theory, pH framework is widely used to model physical systems. 
%A major advantage is that the Hamiltonian can be used as the storage function. 
As a consequence of Remark \ref{energy_s-rlc}, adopting the pH framework (using the total energy as Hamiltonian) does not provide any additional advantage compared to the BM framework. Moreover, s--RLC circuits do not inherit a standard pH structure~\cite{ESCOBAR1999445}:
% Remark \ref{rem:ph_srlc} shows that pH framework yields an complicated representation for s--RLC circuits. Furthermore, in the next we show that the total energy (Hamiltonian) does not result in a useful passivity property.

\Revb{\begin{remark}[Port-Hamiltonian formulation for s--RLC circuits]
		%In passivity theory, port-Hamiltonian framework is widely used to derive new results. 
		Generally, a standard port-Hamiltonian system has the following structure
		\begin{eqnarray}\label{ph}
		\dot{x}=[J(x)-R(x)]\nabla_xH(x)+g(x)u,
		\end{eqnarray}
		where $x:\mathbb{R}_+\rightarrow\mathbb{R}^n$, $u:\mathbb{R}_+\rightarrow\mathbb{R}^m$ denote the state and input, respectively, $J(x)=-J(x)^\top$, $R(x)\geq 0$, $H:\mathbb{R}^n\rightarrow\mathbb{R}_+$ is the Hamiltonian and $g(x)$ the input matrix. The skew symmetric matrix $J(x)$ generally captures the interconnection of the storage elements and $R(x)$ describes the dissipation structure of the system. As a result we have the following dissipation inequality:
		\begin{eqnarray*}
			\dot{H}\leq u^\top y,
		\end{eqnarray*}
		where $y=g(x)^\top \nabla_xH(x)$.	However, s--RLC circuit do not generally inherit this structure. To represent s--RLC circuits it may be needed to modify \eqref{ph} as follows (see  \cite{ESCOBAR1999445} for some examples):
		\begin{eqnarray}\label{ph2}
		\dot{x}=\left[J(x,u)-R(x)\right]\nabla_xH(x) + g(x) V_s,
		\end{eqnarray}
		where $J(x,u)=J_0(x)+\sum_{i=1}^{m}J_i(x) u_i$, $J_i+J_i^\top=0$ for all $i \in \{0, \dots, m\}$, $u$ and $V_s$ denote the duty-ratio and the supply voltage, respectively. This implies the following dissipation inequality:
		\begin{eqnarray}
		\dot{H}\leq V_s^\top y,
		\end{eqnarray}
		which may be not useful for control purpose, since $V_s$ is generally not controllable. 
%		On a key note, \eqref{ph2} can be rewritten as
%		\begin{eqnarray}\label{ph3}
%		\dot{x}=[J_0(x)-R(x)]\nabla_xH+g(x)u_s+G(x)u
%		\end{eqnarray}
%		where the $i^{th}$ column of $G(x)$ is given by $J_i(x)\nabla_xH$. By doing this, the model looses the physical structure of an $RLC$ circuit. 
	\end{remark}}

Alternatively, in \cite[Theorem~1]{1230225} it is shown, under some assumptions, that system~\eqref{BM_non_switched} is passive with respect to the port-variables $u_s$,  $B^\top \dot I$ and the so-called \emph{transformed} mixed-potential function as storage function. However, finding the transformed mixed-potential function is not trivial and often requires that (sufficient) conditions on the system parameters are satisfied. 
Differently, in this work we overcome these issues by proposing a Krasovkii's Lyapunov function similar to \eqref{Krasovskii} as storage function.

    \subsection{Problem Formulation}
    \label{sec:probelm}
    The main goal of this paper is to propose new passivity-based control methodologies for regulating the voltage in RLC and s--RLC circuits.
    
Before formulating the control objective and in order to permit the
controllers design in the next sections, we first make the following
assumption on the available information about systems \eqref{BM_non_switched} and \eqref{BM_switched2}:
\begin{assumption}[Available information]
\label{ass:avail_information} The state variables $I$ and $V$ are
measurable\footnote{Note that, when needed, we also assume that $\dot I$ and $\dot V$ are available.}. 
%The system parameters $L, C, R$ and $G$ are unknown. 
The voltage source $V_s$ in \eqref{BM_switched2} is known and different from zero.
%constant and known.
\end{assumption}

Secondly, in order to formulate the control objective aiming at
voltage regulation, we introduce the following two assumptions on the
existence of a desired reference voltage for both RLC and s--RLC circuits, respectively:
\begin{assumption}[Feasibility for RLC circuits]
\label{ass:feasibility_RLC} There exist a constant desired reference
voltage $V^\star\in\R_+^{\rho}$ and a constant control input $\overline u_s$ such
that a steady state solution $(\overline I, V^\star)$ to system
\eqref{BM_non_switched} satisfies
\begin{align}
\label{BM_non_switched_equilibrium}
\begin{split}
0 &=\Gamma V^\star + R\overline{I} - B\overline{u}_s\\
0 &=\Gamma^\top\overline I - G V^\star.
\end{split}
\end{align}
\end{assumption}
%\begin{assumption}[Existence of unique solution for RLC]\label{ass:unique_rlc_sol}
%	content...
%\end{assumption}
\begin{assumption}[Feasibility for s--RLC circuits]
\label{ass:feasibility_sRLC} There exist a constant desired reference
voltage $V^\star\in\R_+^{\rho}$ and a constant control input $\overline u\in \left(0,1\right)$ such
that a steady state solution $(\overline I, V^\star)$ to system
\eqref{BM_switched2} satisfies
\begin{align}
\label{BM_switched3}
\begin{split}
0 &=\Gamma(\overline u)V^\star + R\overline{I} - B(\overline u)V_s\\
0 &=\Gamma^\top(\overline u)\overline I - G V^\star.
\end{split}
\end{align}
\end{assumption}
%\begin{assumption}[Existence of unique solution for s--RLC]\label{ass:unique_srlc_sol}
%	content...
%\end{assumption}
%
	We notice now that system \eqref{BM_switched2} can be written as follows
	\begin{align*}
	\begin{split}
		\begin{bmatrix}
			-L\dot{I}\\C\dot{V}
		\end{bmatrix}= &\begin{bmatrix}
			{RI} + \Gamma_0V-B_0V_s\\ \Gamma_0^\top I-GV
		\end{bmatrix}\\
		&+\begin{bmatrix}
			\left(\Gamma_1-\Gamma_0\right)V-\left(B_1-B_0\right)V_s\\\left(\Gamma_1-\Gamma_0\right)^\top I
		\end{bmatrix}u.
		\end{split}
	\end{align*}
	Consequently, we introduce the following assumption for controllability purposes:

\begin{assumption}[{Controllability necessary condition}]\label{ass:input Matrix ass}
	There exists {(at least)} an element in the column vector 
%The elements in the column vector 
\begin{equation*}
\begin{bmatrix}
			\left(\Gamma_1-\Gamma_0\right)V-\left(B_1-B_0\right)V_s\\\left(\Gamma_1-\Gamma_0\right)^{\top}I
\end{bmatrix}
\end{equation*}
that is different from zero for all $\left(I,V\right)\in\mathbb{R}^{\sigma\times \rho}$ and any $t\geq0$.
%are not all equal to zero, for all $\left(I,V\right)\in\mathbb{R}^{\sigma\times \rho}$.
\end{assumption}
%The following assumptions will be helpful in presenting a stronger result on the \red{....}
%\begin{assumption}[]\label{ass:asy_stab_RLC}
%	Assume that one of the following conditions hold for RLC circuits \eqref{BM_non_switched}:
%	\begin{itemize}
%		\item[(a)]$P_R(I)>0$ and $P_G(V)>0$.
%		\item[(b)]$P_R(I)>0$ and $P_G(V)+\dfrac{1}{2}||\Gamma V||_G^2>0$.
%		\item[(c)]$P_R(I)+\dfrac{1}{2}||\Gamma^\top I||_R^2>0$ and $P_G(V)>0$.
%	\end{itemize}
%\end{assumption}
%%\begin{lemma}[]
%%Consider system \eqref{BM_non_switched} satisfying Assumptions \ref{label} and \ref
%%\end{lemma}
%%The following assumptions will be helpful in presenting a stronger result on the \red{....}
%\begin{assumption}[]\label{ass:asy_stab_sRLC}
%	Assume that one of the following conditions hold for s--RLC circuits \eqref{BM_switched2}:
%	\begin{itemize}
%		\item[(a)]$P_R(I)>0$ and $P_G(V)>0$.
%		\item[(b)]$P_R(I)>0$ and $P_G(V)+\dfrac{1}{2}||\Gamma(\overline{u}) V||_G^2>0$.
%		\item[(c)]$P_R(I)+\dfrac{1}{2}||\Gamma^\top(\overline{u}) I||_R^2>0$ and $P_G(V)>0$.
%	\end{itemize}
%\end{assumption}
The control objective can now be formulated explicitly:
\begin{obj} [Voltage regulation]
\label{obj:voltage_regulation}
\begin{equation}
\label{eq:voltage_regulation} \lim_{t \rightarrow \infty} V(t) =V^{\star}.
%\overline{V} = V^{\star}.
\end{equation}
\end{obj}

\begin{remark}[Robustness to load uncertainty]
\label{rm:robustness}
In power networks it is generally desired that Objective~\ref{obj:voltage_regulation} is achieved independently from the load parameters, which are indeed often unknown (see also  Assumption~\ref{ass:avail_information}).
\end{remark}
    
\label{sec:model}

                %
                %
                %
                %
                %========================================PROPOSAL
\section{The Proposed Control Approaches}
\label{sec:proposal}
In this section, we present new passivity properties (akin to
differential passivity \cite{6760930}) for the considered RLC circuits~\eqref{BM_non_switched}. Then, we extend these properties to {s--RLC} circuits \eqref{BM_switched2}.
\subsection{New passivity properties}
Novel passive
maps for a class of RLC circuits are presented in
\cite{7846443}\footnote{The class of RLC circuits considered in \cite{7846443} is a sub-class of the systems analyzed in this paper. More precisely, in \cite{7846443} the authors assume that $L, C$ are diagonal and $R, G$ are positive definite. These assumptions are relaxed in this paper (see Assumptions~\ref{ass:LC} and \ref{ass:diss_potential}).}, where the authors use a Krasovskii-type storage 
function similar to \eqref{Krasovskii}, i.e., 
\begin{equation}
\label{com_storage}
S(\dot{I},\dot{V})=\dfrac{1}{2}\|\dot{I}\|_L^2+\dfrac{1}{2}\|\dot{V}\|_C^2.
%\dfrac{1}{2}\dot{I}^\top L\dot{I}+\dfrac{1}{2}\dot{V}^\top C\dot{V}.
\end{equation}
The use of such a storage function enables to relax the constraints on the system parameters required in \cite[Theorem~1]{1230225}. 
%However, this implies that the port-variables depend on the first time derivative of the system states and input. For example, in the RLC circuits, the passive maps is with respect to the first time derivative of the voltage source and the first time derivative of the current. 

Since the storage function \eqref{com_storage} depends on $\dot{I}$ and $\dot{V}$, we consider the following \emph{extended-dynamics}\footnote{These dynamics are differentially extended with respect to time.} of system \eqref{BM_non_switched} 
\begin{subequations}  \label{BM_non_switched_extended}
\begin{align} 
-L\dot{I}&= \Gamma V+RI-Bu_s\label{BM_non_switched_extendeda}\\
C\dot{V}&= \Gamma^\top  I-G V \label{BM_non_switched_extendedb}\\
-L\ddot{I}&= \Gamma \dot{V}+R\dot{I}-B\upsilon_{s}\label{BM_non_switched_extendedc}\\
C\ddot{V}&= \Gamma^\top  \dot{I}-G\dot{V} \label{BM_non_switched_extendedd}\\
\dot{u}_s&=\upsilon_{s},\label{BM_non_switched_extendede}
\end{align}  
\end{subequations}
where $(I,V,\dot{I},\dot{V},u_s)$ and $\upsilon_{s}\in\mathbb{R}^m$ are the (extended) system state and input, respectively. \blue{Then, inspired by \cite{mtns, 7846443, 8431720}, the following result can be established.}
\begin{proposition}[Passivity of RLC circuit]
\label{prop: passivity_Non_switched electrical circuits}
%   Let $G_{ii}=\frac{\partial^2 }{\partial i^2}G, J_{vv}=\frac{\partial^2 }{\partial v^2}J$ be positive semi-definite then we have the following.
Let Assumptions \ref{ass:LC} and \ref{ass:diss_potential} hold.
%Consider system \eqref{BM_non_switched} satisfying assumptions A1-A2.
System \eqref{BM_non_switched_extended} is  passive
with respect to the storage function \eqref{com_storage} and the port-variables $y_s=B^\top\dot{I}$ and $\upsilon_s$.
\end{proposition}
\begin{proof}
The first time derivative of the storage function \eqref{com_storage} along the trajectories of \eqref{BM_non_switched_extended} satisfies
\begin{equation}
\label{Sdot_passivity_RLC}
%\begin{split}
\dot S%&=-\dot{I}^\top\left(\Gamma \dot{V}+R\dot{I}-B\dot{u}_s\right)+\dot{V}^\top\left(\Gamma^\top  \dot{I}-G\dot{V}\right)\\
%&=-\dot{I}^\top R\dot{I}-\dot{V}^\top G\dot{V}+\dot{I}^\top B\dot{u}_s\\
\leq \dot{u}_s^\top B^\top\dot{I}=\upsilon_s^\top y_s.
%\end{split}
\end{equation}
%In the first line of \eqref{Sdot_passivity_RLC}, we have use Assumption \ref{ass:LC} and  substituted the values
%for $L\ddot{I}$ and $C\ddot{V}$ from equation
%\eqref{BM_non_switched_extendedc} and \eqref{BM_non_switched_extendedd} respectively. In the second line, we made use of Assumption \ref{ass:diss_potential}. 
\end{proof}
\redd{\begin{remark}{\bf (Physical interpretation of \eqref{Sdot_passivity_RLC})}
		\label{rm:S_intution}
		The established passivity property can be interpreted as the passivity property derived from the total energy of the \lq dual' circuit, which is constructed by using capacitors as inductors, voltage sources as current sources and vice-versa.
		%Firstly, we note that the storage function \eqref{com_storage_int} represent the total energy stored in the circuit when capacitors becomes inductors and vice-versa. 
		This follows from considering $V_L$ as the voltage across the inductor and $I_C$ as the current through the capacitor. As a consequence, the storage function \eqref{com_storage} can be rewritten as
		\begin{equation}
		\label{com_storage_int}
		S(I_C,V_L)=\dfrac{1}{2}\|V_L\|_{L^{-1}}^2+\dfrac{1}{2}\|I_C\|_{C^{-1}}^2.
		\end{equation}
		In \eqref{com_storage_int}, the term ${1}/{2}\|V_{L}\|^2_{L^{-1}}$ represents the energy stored into a capacitor with capacitance $L^{-1}$ and charge $q_L=L^{-1}V_L$. Similarly, the term ${1}/{2}\|I_C\|^2_{C^{-1}}$ represents the energy stored into an inductor with inductance  $C^{-1}$ and flux $\phi_C=C^{-1}I_C$. 
Furthermore, let $i_s$ denote the current source constructed from a capacitor with capacitance $L^{-1}$ and charge $L^{-1}Bu_s$. As a result, \eqref{Sdot_passivity_RLC} becomes
		\begin{equation}
			\dot{S}\leq  \dot{u}_s^\top B^\top L^{-1}V_L = i_s^\top V_L.
		\end{equation}
% Finally, the proposed passivity property can be interpreted as the property of the dual circuit constructed using 
%\begin{itemize}
%	\item Capacitor $\rightarrow$ inductor,
%	\item Inductor $\rightarrow$ capacitor,
%	\item Voltage source $\rightarrow$ current source,
%	\item Current source $\rightarrow$ voltage source.
%\end{itemize}
\end{remark}}
Before presenting an analogous passive map also for s--RLC circuits~\eqref{BM_switched2}, similarly to \eqref{BM_non_switched_extended} we consider the following extended dynamics of system \eqref{BM_switched2}
\begin{subequations}  \label{BM_switched_extended}
\begin{align} 
-L\dot{I} &=RI + \Gamma(u)V - B(u)V_s\label{BM_switched_extendeda}\\
C\dot V &=\Gamma^\top(u)I - GV\label{BM_switched_extendedb}\\
-L\ddot{I} &=R\dot{I} + \Gamma(u)\dot{V}+\left(\left(\Gamma_1-\Gamma_0\right)V-\left(B_1-B_0\right)V_s\right)\upsilon\label{BM_switched_extendedc}\\
C\ddot V &=\Gamma^\top(u)\dot{I} +\left(\Gamma_1-\Gamma_0\right)^\top I\upsilon- G\dot{V}\label{BM_switched_extendedd}\\
\dot{u}&=\upsilon, \label{BM_switched_extendede}
\end{align}  
\end{subequations}
where $(I,V,\dot{I},\dot{V},u)$ and $\upsilon \in \mathbb{R}$  are the (extended) system state and input, respectively. Then, the following result can be established.
\begin{proposition}[Passivity of s--RLC circuit]
\label{prop: passivity_switched electrical circuits}
Let Assumptions \ref{ass:LC} and \ref{ass:diss_potential} hold. System \eqref{BM_switched_extended} 
is passive with respect to the storage function \eqref{com_storage} and the port-variables $\upsilon$ and
\begin{align}
\label{output_srlc}
\begin{split}
y=&~\Big(\dot{V}^\top \left(\Gamma_1-\Gamma_0 \right)^\top
I-\dot{I}^\top \left(\Gamma_1-\Gamma_0 \right)V\\
&-\dot{I}^\top
\left(B_0-B_1 \right)V_s \Big).
\end{split}
\end{align}
\end{proposition}
\begin{proof}
%Consider the storage function $S$ given in \eqref{com_storage}.\\
The time derivative of the storage function \eqref{com_storage} along the trajectories of \eqref{BM_switched_extended} satisfies
\begin{align}
\label{Sdot_passivity_sRLC}
\begin{split}
\dot S&= -\dot{I}^\top
\Big(\big((1-u)\Gamma_0+u\Gamma_1\big)\dot{V}+\dot{u}(\Gamma_1-\Gamma_0)V\\
&\, \quad+R\dot{I} -\dot{u}(B_1-B_0)V_s\Big)+\dot{V}^\top \Big(\big((1-u)\Gamma_0\\
&\, \quad+u\Gamma_1\big)\dot{I}+\dot{u}(\Gamma_1-\Gamma_0)^\top I-G \dot{V}\Big)\\
&= -\dot{I}^\top R \dot{I}-\dot{V}^\top G \dot{V}+\dot{u}y\\
&\leq  \dot{u}y=\upsilon y.
\end{split}
\end{align}
%In the first line of \eqref{Sdot_passivity_sRLC}, we have use
%Assumption \ref{ass:LC} and  substituted the values for $L\ddot{I}$ and
%$C\ddot{V}$ from equation \eqref{BM_switched_extendedc} and \eqref{BM_switched_extendedd} respectively. In the  second line, we made use of Assumption \ref{ass:diss_potential}. 
\end{proof}

\magenta{Note that, if $V_s$ is controllable, then the storage function \eqref{com_storage}
		along the extended dynamics of \eqref{BM_switched2} satisfies
		\begin{align}
		\label{Sdot_passivity_sRLCa}
		\begin{split}
		\dot S%=& -\dot{I}^\top
		%\Big(\big((1-u)\Gamma_0+u\Gamma_1\big)\dot{V}+\dot{u}(\Gamma_1-\Gamma_0)V+R\dot{I} -\dot{u}(B_1-B_0)V_s-B(u)\dot{V}_s\Big)\\
		%&+\dot{V}^\top \Big(\big((1-u)\Gamma_0+u\Gamma_1\big)\dot{I}+\dot{u}(\Gamma_1-\Gamma_0)I-G \dot{V}\Big)\\
		=& -\dot{I}^\top R \dot{I}-\dot{V}^\top G \dot{V}+\dot{u}y+\dot{I}^\top B(u) \dot{V}_s\\
		\leq &  \upsilon y + \theta^\top  \phi.
		\end{split}
		\end{align}
		where $\theta =\dot{V}_s$ and $\phi=B(u)^\top \dot{I}$.
%		\begin{align}
%		\label{output_srlca}
%		\begin{split}
%		y=&~\Big(\dot{V}^\top \left(\Gamma_1-\Gamma_0 \right)^\top
%		I-\dot{I}^\top \left(\Gamma_1-\Gamma_0 \right)V-\dot{I}^\top
%		\left(B_0-B_1 \right)V_s \Big),\\
%		y_s=&~B(u)^\top \dot{I}
%		\end{split}
%		\end{align}
		Therefore the extended dynamics of \eqref{BM_switched2} are passive with port variables $[\upsilon, ~\theta^\top]^\top$ and $[y,~\phi^\top]^\top$.} %\red{Do we include this in the paper?} 
%\end{remark}}
\begin{remark}{\bf (Insights on the storage function $S$)}
\label{rm:S}
The storage function \eqref{com_storage} depends on the states $\dot{I}, \dot V$ of system~\eqref{BM_non_switched_extended} or~\eqref{BM_switched_extended}. Consequently, $S$ depends on the entire state of the extended system~\eqref{BM_non_switched_extended} or~\eqref{BM_switched_extended}. This follows from replacing $\dot{I}, \dot V$ by the corresponding dynamics \eqref{BM_non_switched_extendeda}--\eqref{BM_non_switched_extendedb} or \eqref{BM_switched_extendeda}--\eqref{BM_switched_extendedb}.
Moreover, we will show in Theorems \ref{prop:Non_Switched_Method_2} and \ref{prop:Switched_Method_2} that designing the controller by using the storage function \eqref{com_storage} enables the achievement of Objective~1 despite the load uncertainty (see Remark~\ref{rm:robustness}). However, the cost of designing a robust controller is the need of information about the first time derivative of the signals $I$ and $V$.  
\end{remark}
By using the passive maps presented in Propositions~\ref{prop: passivity_Non_switched electrical circuits} and~\ref{prop: passivity_switched electrical circuits}, we propose in the next two subsections two different passivity-based control methodologies for both RLC and s--RLC circuits, respectively. % These methodologies rely on the integrability properties of the port-variables. We illustrate these frameworks for both non-switched and switched RLC circuits. In the following, we elaborate these techniques using the average models of standard DC-DC power converters like buck and boost.
\subsection{Output Shaping}
%Using the passive maps presented in Proposition \ref{prop: passivity_Non_switched electrical circuits} and \ref{prop: passivity_switched electrical circuits}, we propose two possible control methodologies.
The first methodology, which we call {\em output shaping}, relies on the integrability property of the \emph{output} port-variable. More precisely, we use the integrated output port-variable to shape the closed-loop storage function. In this subsection, we first extend this methodology to a wider class of RLC circuits than the one considered in \cite{7846443}. Subsequently, we further extend the output shaping methodology to s--RLC circuits. 
\begin{theorem}[Output shaping for RLC circuits]\label{prop:Non_Switched_Method_1}
Let Assumptions~\ref{ass:LC}--\ref{ass:feasibility_RLC} hold. Consider system \eqref{BM_non_switched_extended} with control input $\upsilon_s$ given by %\ref{ass:LC}, \ref{ass:diss_potential}, \ref{ass:avail_information} and \ref{ass:feasibility_RLC}
 %led via $u_s$ with dynamics
                \begin{equation}\label{cont_RLC_output_controller}
               \upsilon_s= \left(\mu_s-k_iB^\top\left( I-\bar I\right)-k_d y_s\right),
                \end{equation}
                with $y_s=B^\top\dot{I}$, $k_d>0,\;k_i>0$ and $\mu_s \in \mathbb{R}^m$. The following statements hold:
                \begin{itemize}
                \item[(a)]  System \eqref{BM_non_switched_extended} in closed-loop with control \eqref{cont_RLC_output_controller} defines a passive map $\mu_s \mapsto y_s$.
                \item[(b)] Let $\mu_s$ be equal to zero. If any of the following conditions holds
                \begin{itemize}
                	\item [(i)]$R>0$ and $G>0$,
                	\item[(ii)]$G> 0$ and $\Gamma^\top$  has full column rank,
                \end{itemize}
                then the solution to the closed-loop system  asymptotically converges to the set% where $\dot{V}=0$, $\dot{I}=0$, $\dot{u}_s=0$ and $B^\top\left(I-\overline{I}\right)=0$.
				%\item [(c)] If $G> 0$ and $\Gamma^\top$  has full column rank,  then the solution to the closed-loop system asymptotically converges to the set \eqref{set:ouput_non_switched}.
                \end{itemize}
                                \begin{equation}\label{set:ouput_non_switched}
                %\begin{split}
              \left\{  \left(I,V,\dot{I},\dot{V},u_s\right): 
               \dot{V}=0,\dot{I}=0,\dot{u}_s=0,B^\top \left(I-\overline{I}\right)=0 \right\}.
                %\end{split}
                 \end{equation}
                \end{theorem}
                \begin{proof}
                	We use the integrated output port-variable to shape the desired closed-loop storage function, i.e.,
                                \begin{equation}\label{com_storage_RLC_clpa}
                S_d= S+\dfrac{1}{2}\left|\left| B^\top (I-\overline I)\right|\right|^2_{k_i},
                \end{equation}
               where $S$ is given by \eqref{com_storage}. 
               Then, $ S_d$ along the trajectories of system~\eqref{BM_non_switched_extended} controlled by \eqref{cont_RLC_output_controller} satisfies
                \begin{subequations}  \label{Sdot_OS_RLC}
                	\begin{align} 
\dot{S}_d%&=-\dot{I}^\top R\dot{I}-\dot{V}^\top G\dot{V}+ \upsilon_s^\top y_s+k_i(I-\overline{I})^\top B y_s\label{Sdot_OS_RLCa}\\
&=-\dot{I}^\top R\dot{I}-\dot{V}^\top G\dot{V}+y_s^\top  \left(\upsilon_s+k_iB^\top (I-\overline{I})\right)\label{Sdot_OS_RLCb}\\
&= -\dot{I}^\top R\dot{I}-\dot{V}^\top G\dot{V}-k_dy_s^\top y_s+\mu_s^\top y_s                 \label{Sdot_OS_RLCc}\\
&\leq \mu_s^\top y_s,\label{Sdot_OS_RLCd}
                	\end{align}  
                \end{subequations}
                where in  \eqref{Sdot_OS_RLCb} we use the controller \eqref{cont_RLC_output_controller}. This concludes the proof of part (a). For part (b-i), let $\mu_s$ be equal to zero. Then, from \eqref{Sdot_OS_RLCc}, there exists a forward invariant set $\Pi$ and by LaSalle's invariance principle the solutions that start in $\Pi$ converge to the largest invariant set contained in 
                \begin{equation}\label{set:forward_inv_set_thm1}
                	\Pi \cap \left\{\left(I,V,\dot{I},\dot{V},u_s\right):\dot{I}=0,\dot{V}=0 \right\}.
                \end{equation}
                From \eqref{BM_non_switched_extendedc} it follows that $B\upsilon_{s}=0$, i.e., $\upsilon_{s}=0$ ($B$ has full column rank). Moreover, from \eqref{cont_RLC_output_controller} it follows that $B^\top (I-\overline I)= 0$, concluding the proof of part (b-i). For part (b-ii), the solutions that start in the forward invariant set $\Pi$ converge to the largest invariant set contained in 
                \begin{equation}\label{set:forward_inv_set_thm1a}
                \Pi \cap \left\{\left(I,V,\dot{I},\dot{V},u_s\right):R\dot{I}=0,\dot{V}=0,y_s=0 \right\}.
                \end{equation}
                On this invariant set, from \eqref{BM_non_switched_extendedd} we obtain $\Gamma^\top\dot{I}= 0$, which implies $\dot{I}= 0$ ($\, \Gamma^\top $ has full column rank). This further implies that, also in this case, the solutions starting in $\Pi$ converge to the set \eqref{set:forward_inv_set_thm1}. The rest of the proof follows from the proof of part (b-i).
                \end{proof}
%                \begin{remark}[Relaxing Assumption \ref{ass:diss_potential}]
%                In the proof of Theorem \ref{prop:Non_Switched_Method_1}, the  asymptomatic stability result is valid even for the case  $R\geq 0$. In this case, from \eqref{Sdot_OS_RLCc} and LaSalle's invariance principle, the solutions that start in the forward invariant set $\Pi$ approach the largest invariant set contained in 
%                \begin{eqnarray}
%                \Pi \cap \left\{\left(I,V,\dot{I},\dot{V},u_s\right):y_s=0,\dot{V}=0, R\dot{I}=0 \right\}.
%                \end{eqnarray}
%                %
%               % we can conclude that $\dot{V}(t)= 0, ~ y_s= 0$. 
%                Then, from \eqref{BM_non_switched_extendedd} we obtain $\Gamma\dot{I}= 0$, which implies $\dot{I}= 0$ ($\Gamma$ has full column rank). This further implies that the solutions that start in $\Pi$ also approach the set \eqref{set:forward_inv_set_thm1}.
%                \end{remark}
                \begin{remark}[Alternative controller to \eqref{cont_RLC_output_controller}]
                The controller \eqref{cont_RLC_output_controller} needs the information of the first time derivative of the inductor current. This can be avoided by rewriting \eqref{cont_RLC_output_controller} as follows:
              %  We define the mapping $u_s:\mathbb{R}^n\rightarrow \mathbb{R}^n$
                \begin{align}
                \label{cont_RLC_b}
                \begin{split}
                u_s&= -\left(k_i\phi+k_d B^\top I\right)\\
                \dot{\phi} &= - \frac{1}{k_i}\mu_s + B^\top \left(I-\overline I\right).
                \end{split}
                \end{align}
 By using  the storage function \eqref{com_storage_RLC_clpa}, the same results of Theorem \ref{prop:Non_Switched_Method_1} can be established analogously. Moreover, note that \eqref{cont_RLC_output_controller} can be rewritten as in \eqref{cont_RLC_b} because of the integrability of the port-variables.
                \end{remark}
We now extend this methodology to s--RLC circuits
\eqref{BM_switched2}. One possible issue in extending this methodology to
  s--RLC circuits may be the integrability of the output
port-variable $y$ given by \eqref{output_srlc}. In order to avoid this issue, we introduce the following assumption: 
\begin{assumption}[Integrating factor]
\label{ass:integrability_m}
There exist $m:\mathbb{R}^{\sigma}\times \mathbb{R}^{\rho}\rightarrow \mathbb{R}$ different from zero and $\gamma:\mathbb{R}^{\sigma}\times \mathbb{R}^{\rho}\rightarrow \mathbb{R}$ such that $\dot{\gamma}=my$.
                \end{assumption}
                It is however worth to mention that the second methodology  (i.e.,  \emph{input shaping}) that we propose in Section \ref{sec:inputshaping} does not need Assumption~\ref{ass:integrability_m}.
Relying on Assumption~\ref{ass:integrability_m}, the following lemma provides a new passive map with integrable output port-variable for system \eqref{BM_switched_extended}.
            \begin{lemma}[Integrable output]
	\label{lemma: passivity_switched electrical circuits}
	Let Assumptions \ref{ass:LC}, \ref{ass:diss_potential} and \ref{ass:integrability_m} hold. System \eqref{BM_switched_extended}
	is passive with port-variables $\dfrac{\upsilon}{m}$ and
	$\dot{\gamma}=my$.
            \end{lemma}
        \begin{proof}
After multiplying and dividing the last line of \eqref{Sdot_passivity_sRLC} by $m$, we obtain
\begin{equation*}
\dot{S}\leq \upsilon y
= \dfrac{\upsilon}{m}\dot{\gamma}.
\end{equation*}
        \end{proof}
         
\begin{theorem}[Output shaping for s--RLC circuits]\label{prop:Switched_Method_1}
                Let Assumptions \ref{ass:LC}--\ref{ass:avail_information} and \ref{ass:feasibility_sRLC}--\ref{ass:integrability_m} hold. Consider system \eqref{BM_switched_extended}  with control    input $\upsilon$ given by
                %\ref{ass:LC}, \ref{ass:diss_potential}, \ref{ass:avail_information}, \ref{ass:feasibility_sRLC}, \ref{ass:input Matrix ass}  and \ref{ass:integrability_m}
                \begin{equation}\label{cont_Switched_RLC_a}
                \upsilon= m\left(\mu-k_i\left(\gamma-\gamma^\star\right)-k_d \dot{\gamma}\right),
                \end{equation}
                with $\gamma^\star=\gamma(\overline{I},V^\star)$,  $k_d>0, k_i >0$ and  $\mu \in \mathbb{R}$. The following statements hold:
                \begin{itemize}
                	\item[(a)] System \eqref{BM_switched_extended} in closed-loop with control \eqref{cont_Switched_RLC_a} defines a passive map $\mu \mapsto \dot{\gamma}$.
                	\item[(b)] Let $\mu$ be equal to zero. If any of the following conditions holds
                	\begin{itemize}
                		\item[(i)]$R>0$ and $G>0$
                		\item[(ii)]$G> 0$, $\Gamma^\top (u)$ has full column rank, and
                		\begin{equation}\label{switched_outshaping_asy_stab_cond_2}
                		\left(\Gamma_1-\Gamma_0\right)V-\left(B_1-B_0\right)V_s \neq 0,
                		\end{equation}
                	\end{itemize}
                	 then the solution to the closed-loop system  asymptotically converges to the set
                	\begin{equation} \label{set:ouput_switched}
                	\left\{\left(I,V,\dot{I},\dot{V},u\right)|~\dot{V}=0,\dot{I}=0,\dot{u}=0,\gamma=\gamma^\star\right\}.
                	\end{equation}
%                	\item[(c)]Furthermore, for $G> 0$ and $\Gamma^\top (u)$ has full column rank, then the solution to the closed-loop system asymptotically converges to the set \eqref{set:ouput_switched}, if
%                	\begin{eqnarray}\label{switched_outshaping_asy_stab_cond_2}
%                	\left(\Gamma_1-\Gamma_0\right)V-\left(B_1-B_0\right)V_s \neq 0.
%                	\end{eqnarray}
                \end{itemize}
                \end{theorem}
                \begin{proof}
 We use the integrated output port-variable $\gamma$ (see Lemma \ref{lemma: passivity_switched electrical circuits}) to shape the desired closed-loop storage function, i.e.,
\begin{equation}\label{storage_output_shaping_clp}
S_d= S+\dfrac{1}{2}k_i\left(\gamma-\gamma^\star\right)^2,
\end{equation} 
where $S$ is given by \eqref{com_storage}. 
Then, $ S_d$ along the trajectories of  system \eqref{BM_switched_extended} controlled by \eqref{cont_Switched_RLC_a} satisfies
                                \begin{subequations}  \label{Sdot_OS_sRLC}
                	\begin{align} 
                \dot{S}_d%&= -\dot{I}^\top R \dot{I}-\dot{V}^\top G \dot{V}+ \upsilon^\top y+k_i(\gamma+a)^\top  \dot{\gamma}\\
&= -\dot{I}^\top R \dot{I}-\dot{V}^\top G \dot{V}+ \dfrac{\upsilon}{m} \dot{\gamma}+k_i\left(\gamma-\gamma^\star\right)  \dot{\gamma}\label{Sdot_OS_sRLCa}\\
%&= -\dot{I}^\top R \dot{I}-\dot{V}^\top G \dot{V}+\dot{\gamma}\left(\dfrac{\upsilon}{m(I,V)}+k_I\left(\gamma-\gamma^\star\right)\right)\label{Sdot_OS_sRLCb}\\
&= -\dot{I}^\top R \dot{I}-\dot{V}^\top G \dot{V} -k_d\dot{\gamma}^2+\mu \dot{\gamma}\label{Sdot_OS_sRLCc}\\
&\leq \mu \dot{\gamma},
                	\end{align}  
                \end{subequations}
            where in \eqref{Sdot_OS_sRLCa} we use Proposition \ref{prop: passivity_switched electrical circuits}, Lemma \ref{lemma: passivity_switched electrical circuits} and the controller~\eqref{cont_Switched_RLC_a}. This concludes the proof of part (a). For part (b-i), let $\mu$ be equal to zero. Then, from \eqref{Sdot_OS_sRLCc}, there exists a forward invariant set $\Pi$ and by LaSalle's invariance principle the solutions that start in $\Pi$ converge to the largest invariant set contained in 
            \begin{equation}\label{set:forward_inv_set_thm2}
            \Pi \cap \left\{\left(I,V,\dot{I},\dot{V},u_s\right):\dot{I}=0,\dot{V}=0,\dot{\gamma}=0 \right\}.
            \end{equation}
            On this invariant set, from \eqref{BM_switched_extendedc} and \eqref{BM_switched_extendedd} it follows that
            \begin{equation*}
            	\begin{bmatrix}
			\left(\Gamma_1-\Gamma_0\right)V-\left(B_1-B_0\right)V_s\\\left(\Gamma_1-\Gamma_0\right)^{ \top}I
            	\end{bmatrix}\upsilon=0.
            \end{equation*}
            Then, according to Assumption \ref{ass:input Matrix ass}, we have $\upsilon=0$, which implies $\dot{u}=0$. Moreover, from \eqref{cont_Switched_RLC_a} it follows that  $\gamma=\gamma^\star$, concluding the proof of part (b-i). For part (b-ii), when only $G$ is positive definite, the solutions that start in the forward invariant set $\Pi$ converge to the largest invariant set contained in 
            \begin{equation}\label{set:forward_inv_set_thm2a}
            \Pi \cap \left\{\left(I,V,\dot{I},\dot{V},u_s\right):R\dot{I}=0,\dot{V}=0,\dot{\gamma}=0 \right\}.
            \end{equation}
            On this invariant set, from \eqref{BM_switched_extendedc} and \eqref{switched_outshaping_asy_stab_cond_2} we obtain $\upsilon=0$. Consequently, from \eqref{BM_switched_extendedd} we have $\Gamma^\top(u)\dot{I}= 0$, which implies $\dot{I}= 0$ ($\,\Gamma^\top(u) $ has full column rank). This further implies that, also in this case, the solutions starting in $\Pi$ converge to the set \eqref{set:forward_inv_set_thm2}. The rest of the proof follows from the proof of part (b-i).
%            
%            By using $\dot{\gamma}=m(I,V)y$, where $y$ is defined in \eqref{output_srlc} and $m(I,V)\neq 0$, it yields 
%%But $\dot{\gamma} = 0$ implies $y= 0$ ($y$ is defined in \eqref{output_srlc}) . Finally, using $\dot{V}= 0$ in $y$ results in
%\begin{align*}
%\begin{split}
%y &= 0\\
% (-\dot{I}^\top (\Gamma_1-\Gamma_0)V-\dot{I}^\top (B_0-B_1)V_s) &=0\\
%-((\Gamma_1-\Gamma_0)V-(B_1-B_0)V_s)^\top\dot{I}&=0.
%\end{split}
%\end{align*}
%Finally, using \eqref{switched_outshaping_asy_stab_cond_2}, we have
%$\dot{I}= 0$. This further implies that the solutions that start in $\Pi$ approach the set \eqref{set:forward_inv_set_thm2}. The rest of the proof follows from the proof of part (b) of the theorem.
                \end{proof}
                %
                %
                %
                %
                %
                %
                %
                %
%From the authors point of view, there are two caveats for this
%methodology. The first one with this methodology is that, in many
%cases, the desired operating point is not directly specified by
%output, but indirectly. For example, in Objective 1, the main
%objective is to regulate the voltage. But if the output is a
%function of current, one need to do an indirect control of voltage.
%We have to control current $I$ to $\bar{I}$ such that $V$ goes to
%$V^\star$. Where $\bar{I}$ is solved from the second equation of
%\eqref{BM_switched3}. This requires the information of the system
%parameters such as load $G$, which is usually unknown  (or partially
%known). The second issue is that the whole methodology relies on
%finding an output that is integrable, which may not always be the
%case. 
\begin{remark}[{Output shaping stability}]
	{Theorems \ref{prop:Non_Switched_Method_1} and \ref{prop:Switched_Method_1} imply that the integrated output port-variables converge to the corresponding desired values and the first time derivatives of the state  converge to zero. \green{However, this generally does not imply that the trajectories of the closed-loop system asymptotically converge to the corresponding desired operating point\footnote{\green{If $\sigma=m$, the input matrix becomes a full rank matrix and, as a consequence, in case of RLC circuits, asymptotic convergence to the corresponding desired operating point can be proved.}}.}
	Furthermore, for the buck, boost, buck-boost and C\'uk applications (see Section \ref{sec:controllers} and Appendix), we will show that Theorems \ref{prop:Non_Switched_Method_1} and \ref{prop:Switched_Method_1} also imply that all the trajectories of the closed-loop system asymptotically converge to the corresponding desired operating point. We also notice that for the input shaping methodology that we present in next subsection, under some mild and reasonable assumptions, the stability results will be strengthened.}
\end{remark}
\begin{remark}[Limitations of output shaping]\label{rem: prob_out_shaping}
	Note that, if the resistance $R$ of a RLC circuit is negligible, then, in order to establish the stability results presented in Theorem \ref{prop:Switched_Method_1} part (b-ii), condition   \eqref{switched_outshaping_asy_stab_cond_2} needs to be satisfied. More specifically, for a buck converter (see Section \ref{subsec::buck}), this is equivalent to have  $V_s\neq 0$, which is in practice true (see also Assumption~\ref{ass:avail_information}). Yet, for a boost converter (see Section \ref{subsec:boost}), satisfying condition~\eqref{switched_outshaping_asy_stab_cond_2} is equivalent to require $V\neq 0$, which generally could be not always true. Moreover, the output shaping control methodology relies on finding $\gamma$ satisfying $\dot\gamma=my$, with $m\neq 0$. This may not always be possible. Finally, designing a controller based on the output shaping methodology requires the information of $\overline I$, which often depends on the load parameters. Consequently, the output shaping methodology is sensitive to load uncertainty (see Remark \ref{rm:robustness}).
%	
%	However, in the examples presented in the later sections, this condition is automatically satisfied. One other caveat in this framework is that, the whole methodology relies on finding an $\gamma$ satisfying $\dot\gamma=m(I,V)y$, that is we are assuming the existence of $m(I,V)$, which may not always be the case.
\end{remark}
\subsection{Input Shaping}\label{sec:inputshaping}
The second methodology, which we call {\em input shaping}, relies on the integrability property of the \emph{input} port-variables $\upsilon_s$ and $\upsilon$ (see Proposition \ref{prop: passivity_Non_switched electrical
circuits} and Proposition \ref{prop: passivity_switched electrical
circuits}), respectively.
Similarly to the output shaping technique, we use the
integrated input port-variable to shape the closed-loop storage
function such that it has a minimum at the desired operating point (see Objective 1). 
Compared to the output shaping methodology, the input shaping methodology has the following advantages:
\begin{itemize}
\item[(i)]  Assumption \ref{ass:integrability_m} on the integrability
 of the output port-variable is no longer needed;
\item[(ii)]  the knowledge of
$\bar{u}_s$ and $\bar{u}$, given by \eqref{BM_non_switched_equilibrium} and \eqref{BM_switched3}, respectively, does not usually require the information of the load parameters (see the examples in Subsections \ref{subsec::buck} and \ref{subsec:boost}), making the input shaping control methodology robust with
respect to load uncertainty;% (see Remark \ref{rm:robustness});
\item[(iii)] condition \eqref{switched_outshaping_asy_stab_cond_2} is not required anymore and, in addition, {all} the trajectories of the extended system converge to the desired operating point.% $\left(\overline{I},V^\star,\overline{u}_s\right)$.
% and Assumption \ref{ass:diss_potential} can be relaxed, i.e.,  `$G> 0$ \emph{and} $\Gamma^\top$  has full column rank' becomes `$R>0$ \emph{or} $G>0$'.
\end{itemize}
%\noindent Note that, from now, when we use Assumption \ref{ass:diss_potential} in the input shaping framework, we refer to its relaxed version.

We now first present the input shaping methodology for RLC circuits \eqref{BM_non_switched}.% Then, we extend it to the case of s--RLC circuits \eqref{BM_switched2}.
%We first propose this methodology for extended dynamics of RLC circuits in \eqref{BM_non_switched_extended} and thereafter extend it to the extended dynamics of s--RLC circuits \eqref{BM_switched_extended}.
\begin{theorem}[Input shaping for RLC circuits]\label{prop:Non_Switched_Method_2}
	Let Assumptions~\ref{ass:LC}--\ref{ass:feasibility_RLC} hold. Consider system \eqref{BM_non_switched_extended} with control input $\upsilon_{s}$ given by
	\begin{equation}\label{cont_RLC_input_controller}
	\upsilon_s=
	\dfrac{1}{k_d}\left(\mu_s-k_i\left(u_s-\bar{u}_s\right)-
	y_s\right),
	\end{equation}
	with $y_s=B^\top\dot{I}$, $k_d>0,\;k_i>0$ and  $\mu_s \in \mathbb{R}^m$. The following statements hold:
	\begin{itemize}
		\item[(a)]System \eqref{BM_non_switched_extended} in closed-loop with control \eqref{cont_RLC_input_controller} defines a passive map $\mu_s \mapsto \dot{u}_s$ (note that $u_s$ is a state of the extended system \eqref{BM_non_switched_extended}).
		\item[(b)] Let $\mu_s$ be equal to zero. If any of the following conditions holds
	\begin{itemize}
\item[(i)] $R>0$ and $G>0$  
\item[(ii)]$R> 0$ and $\Gamma$  has full column rank  
\item[(iii)]$G> 0$ and $\Gamma^\top$  has full column rank,
		\end{itemize}
then the solution to the closed-loop system asymptotically converges to the set
		\begin{equation} \label{set:input_non_switched}
		\left\{\left(I,V,\dot{I},\dot{V},u_s\right): \dot{V}=0,\dot{I}=0,\dot{u}_s=0,u_s=\overline{u}_s\right\}.
		\end{equation}
		\item[(c)] {If any of the conditions in (b) holds} and the matrix 
		\begin{equation} \label{sys_matrix_RLC}
		\mathcal{A}_s	=\begin{bmatrix}
			R & \Gamma\\\Gamma^\top & -G
			\end{bmatrix}
		\end{equation}
has full-rank, then the solution to the closed-loop system asymptotically converges to the desired  operating point {$\left(\overline{I},V^\star,0,0,\overline{u}_s\right)$}, which is \green{unique}.
	\end{itemize}
                \end{theorem}
                \begin{proof}
 We use the integrated input port-variable to shape the desired closed-loop storage function, i.e,
                	\begin{equation}\label{com_storage_RLC_clp}
                	S_d= S+\dfrac{1}{2}\left|\left| u_s-\bar{u}_s\right|\right|_{k_i}^2,
                	\end{equation}
                	where $S$ is given by \eqref{com_storage}. 
                	Then, $ S_d$ along the trajectories of system \eqref{BM_non_switched_extended} controlled by \eqref{cont_RLC_input_controller} satisfies
                	\begin{subequations}  \label{Sdot_IS_RLC}
                		\begin{align} 
                		\dot{S}_d%&=-\dot{I}^\top R\dot{I}-\dot{V}^\top G\dot{V}+ \upsilon_s^\top y_s+k_i\dot{u}_s^\top \left(u_s-\overline{u}_s\right)\label{Sdot_IS_RLCa}\\
                		&=-\dot{I}^\top R\dot{I}-\dot{V}^\top G\dot{V}+\dot{u}_s^\top \left(y_s+k_i\left(u_s-\overline{u}_s\right)\right)\label{Sdot_IS_RLCb}\\
                		&= -\dot{I}^\top R\dot{I}-\dot{V}^\top G\dot{V}-k_d\dot{u}_s^\top \dot{u}_s+\mu_s^\top \dot{u}_s                 \label{Sdot_IS_RLCc}\\
                		&\leq \mu_s^\top \dot{u}_s,\label{Sdot_IS_RLCd}
                		\end{align}  
                	\end{subequations}
                	where in \eqref{Sdot_IS_RLCb} we use Proposition \ref{prop: passivity_Non_switched electrical circuits} and the controller \eqref{cont_RLC_input_controller}. This concludes the proof of part (a). For part (b-i), let $\mu_s$ be equal to zero. Then, from \eqref{Sdot_IS_RLCc}, there exists a forward invariant set $\Pi$ and by LaSalle's invariance principle the solutions that start in $\Pi$ converge to the largest invariant set contained in 
                	\begin{equation}\label{set:forward_inv_set_thm3}
                	\Pi \cap \left\{\left(I,V,\dot{I},\dot{V},u_s\right):\dot{I}=0,\dot{V}=0,\dot{u}_s=0 \right\}.
                	\end{equation}
                	On this invariant set, $\dot{I}=0$ and $\dot{u}_s=0$ further imply $y_s=0$ and $\upsilon_{s}=0$, respectively. Consequently, from \eqref{cont_RLC_input_controller} it  follows that $u_s=\bar{u}_s$, concluding the proof of part (b-i). For part (b-ii) and (b-iii), the solutions that start in the forward invariant set $\Pi$ converge to the largest invariant set contained in 
                	\begin{equation}\label{set:forward_inv_set_thm3a}
                	\Pi \cap \left\{\left(I,V,\dot{I},\dot{V},u_s\right):R\dot{I}=0,G\dot{V}=0,\dot{u}_s=0 \right\}.
                	\end{equation}
                	 On this set, from \eqref{BM_non_switched_extendedc} and \eqref{BM_non_switched_extendedd} we get $\Gamma \dot{V}= 0$ and $\Gamma^\top \dot{I}= 0$, respectively. Consequently, if (b-ii) or (b-iii) holds, then $\dot{I}= 0$ and $\dot{V}= 0$. This further implies that the solutions that start in $\Pi$ converge to the set \eqref{set:forward_inv_set_thm3}. The rest of the proof follows from the proof of part (b). \green{For part (c), we first notice that from \eqref{BM_non_switched_equilibrium} we have
                	\begin{align}
                	 \label{BM_non_switched_equilibrium_thm3a}
                	 \begin{split}
                	 \begin{bmatrix}
                	 \overline I\\V^\star
                	 \end{bmatrix}&=\mathcal{A}_s^{-1}\begin{bmatrix}
									B\overline{u}_s\\0
                	 				\end{bmatrix},
                	 \end{split}
                	 \end{align}
 implying that $(\overline I, V^\star)$ is unique.              	 
                	Moreover, on the set \eqref{set:forward_inv_set_thm3}, from \eqref{BM_non_switched_extendeda} and \eqref{BM_non_switched_extendedb} we obtain
                	 \begin{align}
                	 \label{BM_non_switched_equilibrium_thm3}
                	 \begin{split}
                	 \begin{bmatrix}
                	 I\\V
                	 \end{bmatrix}&=\mathcal{A}_s^{-1}\begin{bmatrix}
									B\overline{u}_s\\0
                	 				\end{bmatrix}.
                	 \end{split}
                	 \end{align}
                	 Then, from \eqref{BM_non_switched_equilibrium_thm3a}, $I$ and $V$ converge to $\overline I$ and $V^\star$, respectively.}
                \end{proof}
We now extend these results to s--RLC circuits \eqref{BM_switched2}. 
\begin{theorem}[Input shaping for s--RLC circuits]\label{prop:Switched_Method_2}
	                Let Assumptions \ref{ass:LC}--\ref{ass:avail_information}, \ref{ass:feasibility_sRLC} and \ref{ass:input Matrix ass} hold. Consider system \eqref{BM_switched_extended}  with control input $\upsilon$ given by   
	\begin{equation}\label{cont_Switched_sRLC_b}
	\upsilon= \dfrac{1}{k_d}\left(\mu-k_i\left(u-\bar u\right)-y\right),
	\end{equation}
	with $y$ given by \eqref{output_srlc}, $k_d>0, k_i >0$ and $\mu \in \mathbb{R}$. The following statements hold:
	\begin{itemize}
		\item[(a)]System \eqref{BM_switched_extended} in closed-loop with control \eqref{cont_Switched_sRLC_b} defines a passive map $\mu \mapsto \dot{u}$  (note that $u$ is a state of the extended system \eqref{BM_switched_extended}).
		\item[(b)] Let $\mu$ be equal to zero. If any of the following conditions holds
		\begin{itemize}
			\item[(i)] $R>0$ and $G>0$  
			\item[(ii)]$R> 0$ and $\Gamma(u)$  has full column rank  
			\item[(iii)]$G> 0$ and $\Gamma^\top(u)$  has full column rank,
		\end{itemize}
		then the solution to the closed-loop system asymptotically converges to the set
		\begin{equation} \label{set:input_switched_IS}
\left\{\left(I,V,\dot{I},\dot{V},u\right)|~\dot{V}=0,\dot{I}=0,\dot{u}=0,u=\overline{u}\right\}.
\end{equation}
		\item[(c)] If any of the conditions in (b) holds and the matrix 
		\begin{equation}
		\mathcal{A}=\begin{bmatrix}\label{sys_matrix_sRLC}
		R & \Gamma({\overline{u}})\\\Gamma^\top({\overline{u}}) & -G
		\end{bmatrix}
		\end{equation}
		has full-rank, then the solution to the closed-loop system asymptotically converges to the desired operating point $\left(\overline{I},V^\star,0,0,\overline{u}\right)$, which is \green{unique}.
%		\item[(b)] Let $\mu$ be equal to zero. If $R>0$ and $G>0$, then the solution to the closed-loop system asymptotically converges to the set
%		\begin{eqnarray} \label{set:input_switched_IS}
%		\left\{\left(I,V,\dot{I},\dot{V},u\right)|~\dot{V}=0,\dot{I}=0,\dot{u}=0,u=\overline{u}\right\}.
%		\end{eqnarray}
%		\item[(c)] If, either $R> 0$ and $\Gamma(u)$  has full column rank or $G> 0$ and $\Gamma^\top(u)$  has full column rank, then the solution to the closed-loop system asymptotically converges to the set \eqref{set:input_non_switched}.%, if $\Gamma(u)$ or $\Gamma^\top(u)$  has full column rank respectively.
		%Furthermore, relaxing Assumption \ref{ass:diss_potential} to $R> 0$ or $G> 0$, the solution to the closed-loop system approaches  asymptotically the set \eqref{set:input_switched_IS}, if $\Gamma(u)$ is full rank.
	\end{itemize}
\end{theorem}
                \begin{proof}
                	We use the integrated input port-variable to shape the desired closed-loop storage function, i.e.,
                	\begin{equation}\label{com_storage_sRLC_clp}
                	S_d= S+\dfrac{1}{2}k_i\left(u-\bar{u}\right)^2,
                	\end{equation}
                	where $S$ is given by \eqref{com_storage}. 
                	Then, by using the storage function \eqref{com_storage_sRLC_clp}, the proof is analogous to that of Theorem \ref{prop:Non_Switched_Method_2}.
\end{proof}
\begin{remark}[{Robustness property of input shaping methodology}]\label{rem:robustness}
%Consider the steady-state equations \eqref{BM_non_switched_equilibrium} and \eqref{BM_switched3} of RLC and s--RLC circuits. 
Note that the controllers \eqref{cont_RLC_input_controller} and \eqref{cont_Switched_sRLC_b} proposed in Theorems \ref{prop:Non_Switched_Method_2} and \ref{prop:Switched_Method_2}, respectively, require information of the desired value of the control input. If $R=0$, from  the first line of \eqref{BM_non_switched_equilibrium} and \eqref{BM_switched3}, it follows that $\overline{u}_s$ and $\overline{u}$ require only information of the desired voltage $V^\star$. This implies that the input shaping methodology is robust with respect to load uncertainty (see Remark~\ref{rm:robustness}). 
\end{remark}
%\Rb{\begin{remark}[{Tuning parameters}]\label{rem:gains}
%In Theorems 1, 3 and Theorems 2, 4, we have proposed Proportional Derivative (PD) - type controllers for input $\dot{u}_s$ and $\dot{u}$, respectively. In output shaping these represent PD-type action on an integrated output i.e., $\gamma$. However, the input to the real plant is $u$ or $u_s$. As a consequence, the proportional action becomes a integral like action, and the derivative action becomes a proportional like action. Therefore in output shaping, $k_i$ and $k_d$ represent tuning parameters for integral and proportional actions on $\gamma$.  In a similar manner, for input shaping technique, $\dfrac{k_i}{k_d}$ and $\dfrac{1}{k_d}$ denote the tuning parameters for integral actions on input $u$ or $u_s$ and proportional actions on $\gamma$.
%\end{remark}}
\Rc{\begin{remark}[Initial conditions for $u_s$ and $u$]\label{rem:init_cond}
The control inputs $u$ and $u_s$ of systems \eqref{BM_non_switched} and \eqref{BM_switched} are states of the extended systems \eqref{BM_non_switched_extended} and  \eqref{BM_switched_extended}, respectively. Moreover, we proved that the closed-loop dynamics of these extended systems are asymptotically stable. Therefore, independently of  the initial conditions of $u$ and $u_s$, the proposed dynamic controllers  stabilize the corresponding closed-loop systems to their desired operating points.
\end{remark}}

Before showing the application of the proposed control methodologies to power converters in the next section, we notice that, under certain assumptions on $\Gamma$, the input shaping methodology allows for $R\geq 0$ or $G\geq 0$. Differently, the output shaping methodology allows only for $R\geq 0$. Furthermore, under certain assumptions on the steady state equations, the  input shaping methodology guarantees that \emph{all} the solutions to the extended system  converge to the desired operating point. %Nevertheless, this is not easy to ascertain for output shaping. %Moreover, in the next subsection, we present an possible extension of the proposed techniques to more generic nonlinear load models.
\section{Application to DC-DC Power Converters}
\label{sec:controllers} 
In this section, we use the control
methodologies proposed in the previous section for regulating the output voltage 
of the most widespread DC-DC power converters: the buck and the boost converters\footnote{Buck and boost converters describe in form and function a large family of DC-DC power converters. Moreover, in Appendices \ref{appendixA} and \ref{appendixB} we also study other two  common types of DC-DC power converters: the buck-boost and C\'uk converters}, respectively. 

% In this section, we illustrate the control
% methodologies presented in the previous section for (possibly)
% robust voltage regulation of buck and boost DC-DC power converters.
% In the appendix, we show that buck-boost and Cuk converters can also
% be regulated robustly using input shaping methodology.
%                \subsection{Parallel RLC circuit}
                \subsection{Buck converter}\label{subsec::buck}
                    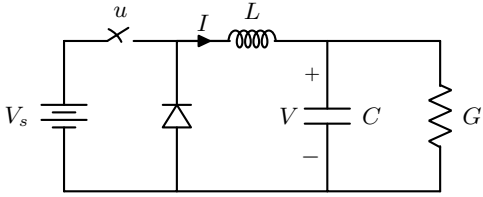
\begin{figure}[t]
        \begin{center}
            \ctikzset{bipoles/length=1cm}
            \begin{circuitikz}[scale=1,transform shape]
                \draw
                (0,1) node [] {} to [cspst=$u$] (1.5,1)
                (1.5,-1) {to [diode, l=, v<={{}}] (1.5,1)}
                (0,-1) node [] {} -- (5.0,-1)
                (0,1) {to [battery, l_=$V_s$] (0,-1)}
                (1.5,1) to [L, l=$L$, i>^={{$I$}}] (3.5,1)
                (3.5,1) {to [C, l=$C$, v={{$V$}}] (3.5,-1)}
                (5.0,-1) {to [R, l_=$G$] (5.0,1)}
                (3.5,1) -- (5.0,1);
                \end{circuitikz}
                \end{center}
                \caption{Electrical scheme of the buck converter.}
                \label{fig:buck_scheme}
                \end{figure}

                %\begin{figure}
                %\begin{center}
                %\ctikzset{bipoles/length=1cm}
                %\begin{circuitikz}[scale=1,transform shape]
                %\draw
                %(0,1) to [cspst=$u$] (1.5,1)
                %(0,-1) node [] {} -- (5,-1)
                %(0,1) {to [battery, l_=$E$] (0,-1)}
                %(1.5,1) {to [L, l=$L$, i>^=$I$] (1.5,-1)}
                %(3.5,1) to [diode] (1.5,1)
                %(3.5,-1) {to [C, l=$C$, v<={{$V$}}] (3.5,1)}
                %(5,-1) {to [R, l_=$G$] (5,1)}
                %(3.5,1) -- (5.0,1);
                %\end{circuitikz}
                %\end{center}
                %\caption{buck-boost}
                %\end{figure}

Consider the electrical scheme of the buck converter in Figure
\ref{fig:buck_scheme}, where the diode is assumed to be ideal. Then,
by applying the Kirchhoff's current (KCL) and voltage (KVL) laws,
the \emph{average} governing dynamic equations of the buck converter
are the following:
                \begin{align}
                \begin{split}
                \label{eq:buck}
                -L\dot{I} &=V - u V_s\\%, \quad u\in(0,1)\\
                C\dot{V} &=I - G V.
                \end{split}
                \end{align}
Equivalently, system \eqref{eq:buck} can be obtained from
\eqref{BM_switched2} with $\Gamma_0=\Gamma_1=1$, $B_0=0$, $B_1=1$ and
$R=0$. By using Proposition \ref{prop: passivity_switched electrical circuits},  the following passivity property is established.
\begin{lemma}[Passivity property of the buck converter]
	Let Assumptions \ref{ass:LC} and \ref{ass:diss_potential} hold. System \eqref{eq:buck} is passive with respect to the storage function \eqref{com_storage} and the port-variables $\dot{u}$ and $\dot{I}V_s$.	
%The system of equations \eqref{eq:buck} describing buck converter dynamics are passive with port-variables $\dot{u}$ and $\dot{I}V_s$.% or more conveniently $E\dot{u}$ and $\dot{I}$.
                \end{lemma}
            % 
              %  {\em Control Objective:} The objective is to maintain a voltage $V^\star$ across the load $G$.
                %{\em Method 1: Control using output shaping}
                By virtue of the above passivity property, we can now use the output shaping and input shaping control methodologies to design voltage controllers. 
%                To use the output shaping methodology, the output port-variable needs to satisfy Assumption \ref{ass:integrability_m}. 
%                In this case, $m(I,V)=1$ yields $\gamma(I,V)=IV_s$. 
                %The following results are now established.
                \begin{corollary}[Output shaping for the buck converter]\label{cor:buck_output}
                Let Assumptions \ref{ass:LC}--\ref{ass:avail_information} and {\ref{ass:feasibility_sRLC}} hold.	Consider system \eqref{eq:buck}  with the dynamic controller 
                	\begin{equation}\label{cont_buck}
                	\dot{u}= -V_s\left(k_i\left(I-\overline{I} \right)+k_d \dot{I}\right),
                	\end{equation}
                	with $k_d>0$ and $k_i>0$. Then, the solution $(I,V,u)$ to the closed-loop system asymptotically converges to the desired steady-state  $\left(\overline{I},V^\star,\overline{u}\right)$.% defined in Assumption \ref{ass:feasibility_sRLC}.% $\left(GV^\star,V^\star,\dfrac{V^\star}{V_s}\right)$.
                \end{corollary}
                \begin{proof}
%First, in order to satisfy Assumption \ref{ass:integrability_m}, we select $m=1$, $\gamma=I V_s$  and $\gamma^\star=\overline{I} V_s$. 
For the buck converter \eqref{eq:buck},  condition \eqref{switched_outshaping_asy_stab_cond_2} is equivalent to require $V_s\neq 0$, which holds by Assumption~\ref{ass:avail_information}. Consequently, Theorem \ref{prop:Switched_Method_1} can  be used by selecting $m=1$, $\gamma=I V_s$  and $\gamma^\star=\overline{I} V_s$. In analogy with Theorem~\ref{prop:Switched_Method_1}, the solutions to the closed-loop system converge to the set
\begin{equation}\label{set:Buck_OS}
\Pi \cap \left\{\left(I,V,u\right):\dot{I}=0,\dot{V}=0\right\}.
\end{equation}
By differentiating the first line of \eqref{eq:buck}, on this invariant set we get $\dot{u}=0$. As a consequence, from \eqref{cont_buck} it follows that $I=\overline{I}$ which further implies $V=V^\star$ and $u=\overline{u}$ (see Assumption~\ref{ass:feasibility_sRLC}). 
                \end{proof}
% \begin{remark}

% Controller dynamics \eqref{cont_buck} yielded from Output Shaping methodology depends on the load parameter $G$. Hence, the methodology is not robust to the uncertainties presented due to Load variations. This can be seen from simulations presented in Figure \ref{fig:ex_buck_output_shaping}. At time $t=1$, the load has been changed from ...to ..... The effect of this can be noted in the Voltage plot.
% \end{remark}
                \begin{corollary}[Input shaping for the buck converter]\label{prop:input_shaping_buck} Let Assumptions \ref{ass:LC}--\ref{ass:avail_information} and  \ref{ass:feasibility_sRLC} hold.
                	                	Consider system \eqref{eq:buck} with the dynamic controller
                \begin{equation}\label{cont_buck_input}
\dot{u}= -\dfrac{1}{k_d}\left(k_i\left(u-\overline{u} \right)+ V_s\dot{I}\right),
\end{equation}
                	with $k_d >0$ and $k_i>0$. Then, the solution $(I,V,u)$ to the closed-loop system asymptotically converges to the desired steady-state  $\left(\overline{I},V^\star,\overline{u}\right)$.% $\left(GV^\star,V^\star,\dfrac{V^\star}{V_s}\right)$.
                \end{corollary}
                \begin{proof}
The proof is analogous to that of Theorem \ref{prop:Switched_Method_2}. 
                \end{proof}
% \begin{remark}
% Unlike the Output Shaping controller \eqref{cont_buck}, the input shaping controller presented in \eqref{cont_buck_input} in independent of load $G$. This makes the controller robust with respect to the uncertainties present in value of the load parameter $G$. Moreover, this can also be inferred from simulation presented in Figure \ref{fig:ex_buck_input_shaping}. At time $t=1$, the load has been changed from .. to.. . From the voltage plot in the figure, we can see that the load voltage $V(t)$ settles back to its desired true value.   
% \end{remark}
%\begin{remark}[Alternative proof to Corollary \ref{prop:input_shaping_buck}]
%Note that the stability results of Corollary \ref{prop:input_shaping_buck} can also be established from analyzing the first time derivative of the storage function \eqref{com_storage_sRLC_clp} along the solutions to the closed-loop system \eqref{eq:buck}, \eqref{cont_buck_input}, i.e.,
%\begin{equation}\label{buck_IS_Sdot}
%\dot{S}_d=-G\dot{V}^2-k_d\dot{u}^2.
%\end{equation}
%Then, there exists a forward invariant set $\Pi$ and by LaSalle's invariance principle the solutions that start in $\Pi$  converge to the largest invariant set contained in 
%\begin{equation}\label{set:Buck_IS}
%\Pi \cap \left\{\left(I,V,u\right):\dot{u}=0,\dot{V}=0\right\}.
%\end{equation}
%By differentiating the second line of \eqref{eq:buck}, on this invariant set   we get $\dot{I}=0$. Moreover, from \eqref{cont_buck_input} it follows that $u=\overline{u}$ which further implies $V=V^\star$ and $I=\overline{I}=GV^{\star}$. 
%\end{remark}
\subsection{Boost converter} \label{subsec:boost}
Consider now the electrical scheme of the boost converter in Figure
\ref{fig:boost_scheme}, where the diode is again assumed to be
ideal. The \emph{average} governing dynamic equations of the boost
converter are the following:
                \begin{align}
                \begin{split}
                \label{eq:boost}
                -L\dot{I} &=(1 - u) V - V_s\\
                C\dot{V} &=(1 - u) I - G V.
                \end{split}
                \end{align}
Also in this case, system \eqref{eq:boost} can be 
obtained from \eqref{BM_switched2} with  $\Gamma_0=1$, $\Gamma_1=0$,
$B_0=B_1=1$ and $R=0$. By using Proposition \ref{prop: passivity_switched electrical circuits}, the following passivity property is established.
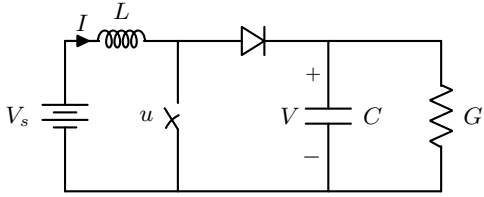
\begin{figure}[t]
                \begin{center}
                \ctikzset{bipoles/length=1cm}
                \begin{circuitikz}[scale=1,transform shape]
                \draw
                %(0,1) node [] {} to [R, l=$R_t$, i>^=$I_t$] (2,1)
                (0,1) to [L, l=$L$, i>^=$I$] (1.5,1)
                (0,-1) node [] {} -- (5,-1)
                (0,1) {to [battery, l_=$V_s$] (0,-1)}
                (1.5,-1) {to [cspst=$u$] (1.5,1)}
                (1.5,1) to [diode] (3.5,1)
                (3.5,1) {to [C, l=$C$, v>={{$V$}}] (3.5,-1)}
                (5,-1) {to [R, l_=$G$] (5,1)}
                (3.5,1) -- (5.0,1);
                \end{circuitikz}
                \end{center}
                \caption{Electrical scheme of the boost converter.}
                \label{fig:boost_scheme}
                \end{figure}
                \begin{lemma}[Passivity property of boost converter]
               Let Assumptions \ref{ass:LC} and \ref{ass:diss_potential} hold. System \eqref{eq:boost} is passive with respect to the storage function \eqref{com_storage} and the port-variables $\dot{u}$ and $\dot{I}V-\dot{V}I$.
                \end{lemma}
                % \begin{proof}
                % Consider the following storage function:
                % \begin{eqnarray}\label{buck_storage_function}
                % S&=& \dfrac{1}{2}L\left(\dfrac{di}{dt}\right)^2+\dfrac{1}{2}C\left(\dfrac{dv}{dt}\right)^2
                % \end{eqnarray}
                % The time differential of $S$ along the trajectories of \eqref{boost_converter} is
                % \begin{eqnarray*}
                % \dfrac{d}{dt}S&=&-\dfrac{di}{dt}\left((1-u)\dfrac{dv}{dt}-\dfrac{du}{dt}v\right)\\&&+C\dfrac{dv}{dt}\left((1-u)\dfrac{di}{dt}-\dfrac{du}{dt}i-G_l\dfrac{dv}{dt}\right)\\
                % &=&-G_l\left(\dfrac{dv}{dt}\right)^2+\dfrac{du}{dt}\dfrac{di}{dt}\\
                % &\leq& \dfrac{du}{dt}\left(\dfrac{di}{dt}v-\dfrac{dv}{dt}i\right).
                % \end{eqnarray*}
                % \end{proof}

\begin{remark}[Integrable output port-variables for the boost converter]
Note that the output port-variable $\dot{I}V-\dot{V}I$ is not integrable. It is however possible to find a different output port-variable  that is indeed integrable (see Lemma \ref{lemma: passivity_switched electrical circuits}). % \ref{lemm:m(x)}. 
More precisely, if we choose for instance $m={1}/{I^2}$, we obtain the passive map $\dot{u}I^2 \mapsto -\frac{d}{dt}({V}/{I})$ (see Table \ref{table:passive maps} for different passivity properties corresponding to different choices of (integrable) output port-variables).
%                 \begin{eqnarray*}
%                 \dfrac{d}{dt}S&\leq&\dfrac{du}{dt}\left(\dfrac{di}{dt}v-\dfrac{dv}{dt}i\right)\\
%                 &=& \dfrac{du}{dt}\left(\dfrac{di}{dt}v-\dfrac{dv}{dt}i\dfrac{v^2}{i^2}\right)\\
%                 &=& \left(\dfrac{du}{dt}v^2\right)\dfrac{d}{dt}\left(\dfrac{i}{v}\right)\\
%                 &=&\left(-\dfrac{du}{dt}i^2\right)\dfrac{d}{dt}\left(\dfrac{v}{i}\right)
%                 \end{eqnarray*}
%                 where $y=\dfrac{d}{dt}\dfrac{i}{v}$.
                %See Table \ref{table:passive maps} for more passivity properties corresponding to different integrable output port-variables.

\begin{table}[t]
                \centering
                \caption{Passive maps for the boost-converter}
                { \begin{tabular}{c  c c}
 \toprule
 $m(I,V)$  & Passive map & $\gamma(I,V)$ \\ [0.5ex]
\midrule
 1& $\dot{u} \mapsto \dot{I}V-\dot{V}I$ & \\
 \midrule
 $\dfrac{1}{V^2}$& $V^2\dot{u} \mapsto \dfrac{d}{dt}\dfrac{I}{V}$ & $\dfrac{I}{V}$\\
 \midrule
 $\dfrac{1}{I^2}$& $I^2\dot{u} \mapsto -\dfrac{d}{dt}\dfrac{V}{I}$ & $-\dfrac{V}{I}$\\
 \midrule
 $\dfrac{1}{V^2+I^2}$& $(V^2+I^2)\dot{u} \mapsto \dfrac{d}{dt}\tan^{-1}\left(\dfrac{I}{V}\right)$& $\tan^{-1}\left(\dfrac{I}{V}\right)$ \\
 \midrule
 $\dfrac{1}{IV}$ &  $(IV)\dot{u} \mapsto \dfrac{d}{dt}ln\left(\dfrac{I}{V}\right)$& $ln\left(\dfrac{I}{V}\right)$ \\  [1ex]
 \bottomrule
\end{tabular}}
                \label{table:passive maps}
                \end{table}

%                 This gives use the following passive maps:
%                 \begin{itemize}
%                 \item[(i)] $\dfrac{du}{dt}\rightarrow\left(\dfrac{di}{dt}v-\dfrac{dv}{dt}i\right)$,
%                 \item[(ii)] Choosing $m(i,v)=\dfrac{1}{v^2}$ gives $\dfrac{du}{dt}v^2\rightarrow \dfrac{d}{dt}\dfrac{i}{v}$,
%                 \item[(iii)]Choosing $m(i,v)=\dfrac{1}{i^2}$ gives $\dfrac{du}{dt}i^2\rightarrow -\dfrac{d}{dt}\dfrac{v}{i}$,
%                 \item[(iv)]Choosing $m(i,v)=\dfrac{1}{v^2+i^2}$ gives $\dfrac{du}{dt}\left(i^2+v^2\right)\rightarrow \dfrac{d}{dt}\tan^{-1}\left(\dfrac{i}{v}\right)$,
%                 \item[(v)]Choosing $m(i,v)=\dfrac{1}{iv}$ gives $\dfrac{du}{dt}\left(iv\right)\rightarrow \dfrac{d}{dt}ln\left(\dfrac{i}{v}\right)$.
%                 \end{itemize}
                \end{remark}
%                 \noindent{\em Control Objective:} The objective is to maintain a voltage $v^\star$ across the load $G_l$.
%                 We achieve the control objective using Choosing $m(i,v)=\dfrac{1}{i^2}$.
%                 %{\em Method 1: Control using output shaping}
  By virtue of the above passivity property, we can now use the output shaping and input shaping control methodologies to design voltage controllers.                
                \begin{corollary}[Output shaping for the boost converter] 
                \label{corollary3}
                Let Assumptions \ref{ass:LC}--\ref{ass:avail_information} and \ref{ass:feasibility_sRLC} hold.  Moreover, let $V(t)$ be different from zero for any $t\geq 0$.
                	 Consider system \eqref{eq:buck} with the dynamic controller
                \begin{equation}\label{cont_boost}
\dot{u}= -\dfrac{1}{V^2}\left(k_i\left(\dfrac{I}{V}-\dfrac{\overline{I}}{V^\star} \right)+k_d \dfrac{d}{dt}\dfrac{I}{V}\right),
\end{equation}
                	with $k_d>0$ and $k_i>0$. Then, the solution $(I,V,u)$ to the closed-loop system asymptotically converges to the desired steady-state  $\left(\overline{I},V^\star,\overline{u}\right)$.
                \end{corollary}
                \begin{proof}
For the boost converter \eqref{eq:boost}, condition \eqref{switched_outshaping_asy_stab_cond_2} is equivalent to require $V(t)\neq 0$ for any $t\geq0$, which holds by assumption.
Consequently, Theorem \ref{prop:Switched_Method_1} can  be used by selecting for instance $m={1}/{V^2}$, $\gamma={I}/{V}$  and $\gamma^\star={\overline{I}}/ {V^\star}$. 
In analogy with Theorem~\ref{prop:Switched_Method_1}, the solutions to the closed-loop system converge to the set
	\begin{equation}\label{set:Boost_OS}
	\Pi \cap \left\{\left(I,V,u\right):\dot{V}=0,\dot{I}=0\right\}.
	\end{equation}
 By differentiating the first line of \eqref{eq:boost}, on this invariant set  we get $\dot{u}=0$. As a consequence, from \eqref{cont_boost} it follows that $\gamma=\gamma^\star$.  Then, from the second line of \eqref{eq:boost} it yields
	\begin{equation}
		u=1-G\dfrac{V}{I}=1-G\dfrac{1}{\gamma}=1-G\dfrac{V^\star}{\overline{I}}=\overline{u},
	\end{equation}
which further implies $V=V^\star$ and $I=\overline{I}$. 
%This implies that the states $(I,V,u)$ of the closed-loop system approach asymptotically to $\left(\overline{I},V^\star,\overline{u}\right)$.
                \end{proof}
\begin{corollary}[Input shaping for boost converter]\label{prop:input_shaping_boost}
	Let Assumptions \ref{ass:LC}--\ref{ass:avail_information},  \ref{ass:feasibility_sRLC} and \ref{ass:input Matrix ass} hold\footnote{For the boost converter, Assumption \ref{ass:input Matrix ass} is equivalent to require that $V$ and $I$ are not equal to zero at the same time (i.e., $V$ can be equal to zero when $I$ is different from zero and vice versa). We note that to use the output shaping methodology we need a stronger assumption, i.e., $V$ different from zero for any $t\geq 0$ (see Corollary \ref{corollary3}).}. Consider system \eqref{eq:boost} with the dynamic controller
\begin{equation}\label{cont_boost_input}
\dot{u}:=-\dfrac{1}{k_d}\left(k_i\left(u-\bar{u} \right)+ \left(\dot{I}V-\dot{V}I\right)\right),
\end{equation}
                with $k_d>0$ and $k_i>0$. Then, the solution $(I,V,u)$ to the closed-loop system asymptotically converges to the desired steady-state  $\left(\overline{I},V^\star,\overline{u}\right)$.
                \end{corollary}
                \begin{proof}
%                Note that $R=0$. To use the input shaping controller \eqref{cont_Switched_sRLC_b}, proposed in Theorem \ref{prop:Switched_Method_2}, we need to show that $\Gamma(u)$ is full rank. 
%                In the case of boost converter $\Gamma_0=1,~\Gamma_1=0$, hence $\Gamma(u)=1-u$. The full rank condition of $\Gamma(u)=1-u$ is satisfied by Assumption \ref{ass:feasibility_sRLC}. The rest follows from the proof of the Theorem \ref{prop:Switched_Method_2}. 
The proof is analogous to that of Theorem \ref{prop:Switched_Method_2}.
\end{proof}

In Table \ref{tab:supplyrates} we have summerized the passivity properties derived in the pH~{\cite{915398}}, BM~{\cite{1230225}} and proposed framework, respectively.

For the sake of completeness, we now show in Figures \ref{fig:ex_buck_output_shaping}--\ref{fig:ex_boost_input_shaping} the simulation results obtained by implementing the proposed methodologies to control the output voltage of a buck and boost converter, respectively.  
In order to verify the robustness property of the proposed controllers with respect to the load uncertainty, the value of the load is changed from $G$ to $G+\Delta G$, with $\Delta G$ uncertain, at the time instant $t=$ \SI{1}{\second} (all the simulation parameters are reported at the end of the caption of each figure). 
More precisely, Figures \ref{fig:ex_buck_output_shaping} and \ref{fig:ex_boost_output_shaping} show that after the load variation the voltage converges to a steady state value different from the desired one. Controllers \eqref{cont_buck} and \eqref{cont_boost} depend indeed on $\overline I = GV^\star$ and, therefore, require  the information of $G$. 
On the contrary, Figures \ref{fig:ex_buck_input_shaping} and \ref{fig:ex_boost_input_shaping} clearly show that the input shaping methodology is robust with respect to load uncertainty (see also Remark \ref{rem:robustness}).
\magenta{Furthermore, for the sake of fairness, we compare the proposed input shaping methodology with the {\em Parallel Damping PBC} approach proposed in \cite[Section V]{1323174}. 
%For best results, we set the value of $\gamma$ in \cite[Section V]{jeltsema2005modeling} to 0.99.
 Figures \ref{fig:ex_buck_output_shaping} and \ref{fig:ex_boost_output_shaping} indicate that {\em Parallel Damping PBC} approach is also robust with respect to load variation. However, it is important to note that the {\em Parallel Damping PBC} approach requires the information of the filter inductance $L$ and capacitance $C$.} 
 \begin{figure}
 	\includegraphics[width=\columnwidth]{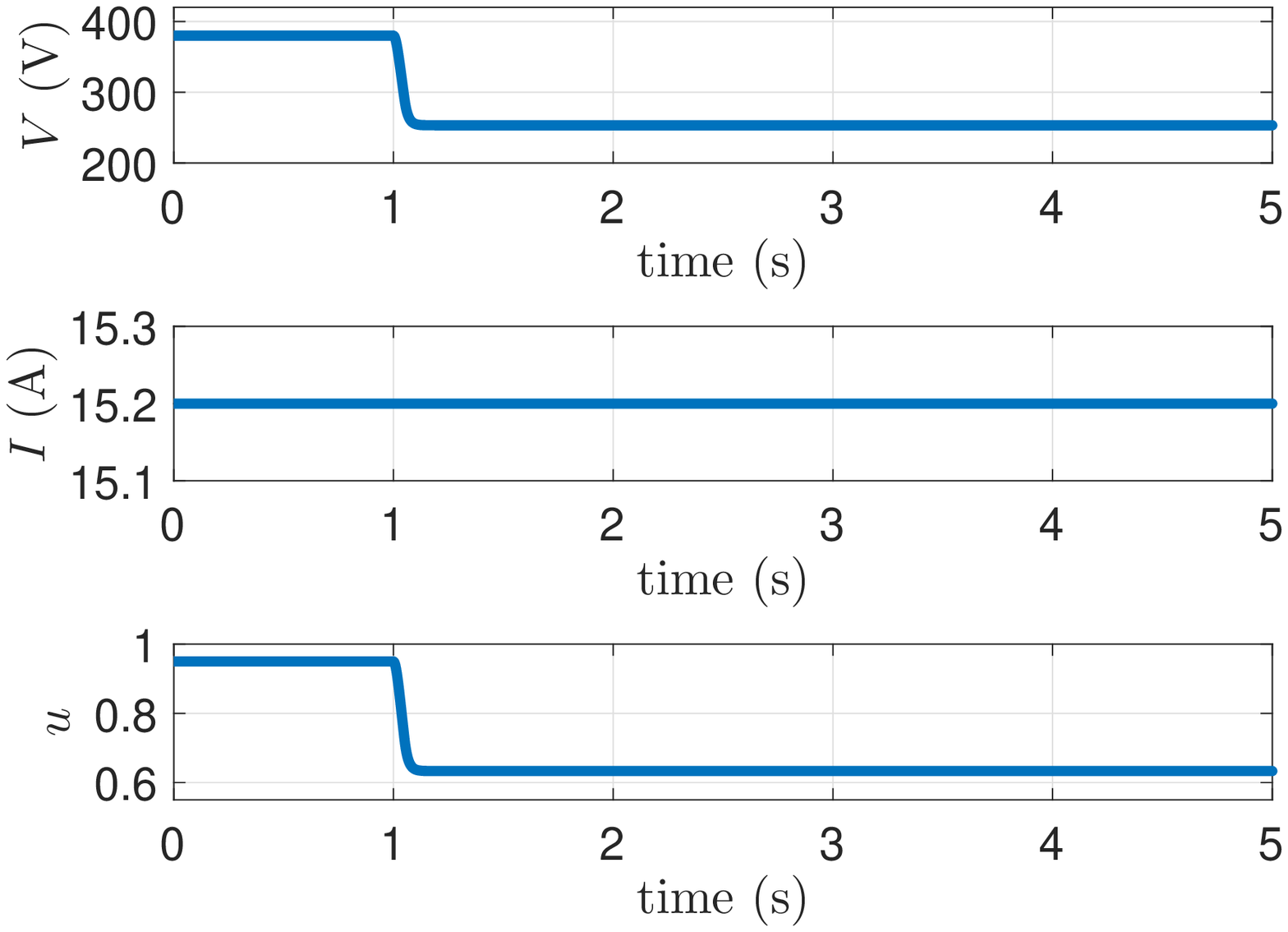}
 	\caption{(Output shaping for the buck converter) From the top: time evolution of the voltage, current and duty cycle considering a load variation $\Delta G$ at the time instant $t=$ 1 s (Parameters: $L =$ \SI{1}{\milli\henry}, $C =$ \SI{1}{\milli\farad}, $V_s =$ \SI{400}{\volt}, $G =$ \SI{0.04}{\siemens}, $\Delta G =$ \SI{0.02}{\siemens}, $V^\star =$ \SI{380}{\volt}, $k_d =$ \num{5e5}, $k_i =$ \num{1e7}).} \label{fig:ex_buck_output_shaping}
 \end{figure}
%%%%%%%%%%%%%%%%%%%%%%%% FIGURE buck_input_shaping
\begin{figure}
\includegraphics[width=\columnwidth]{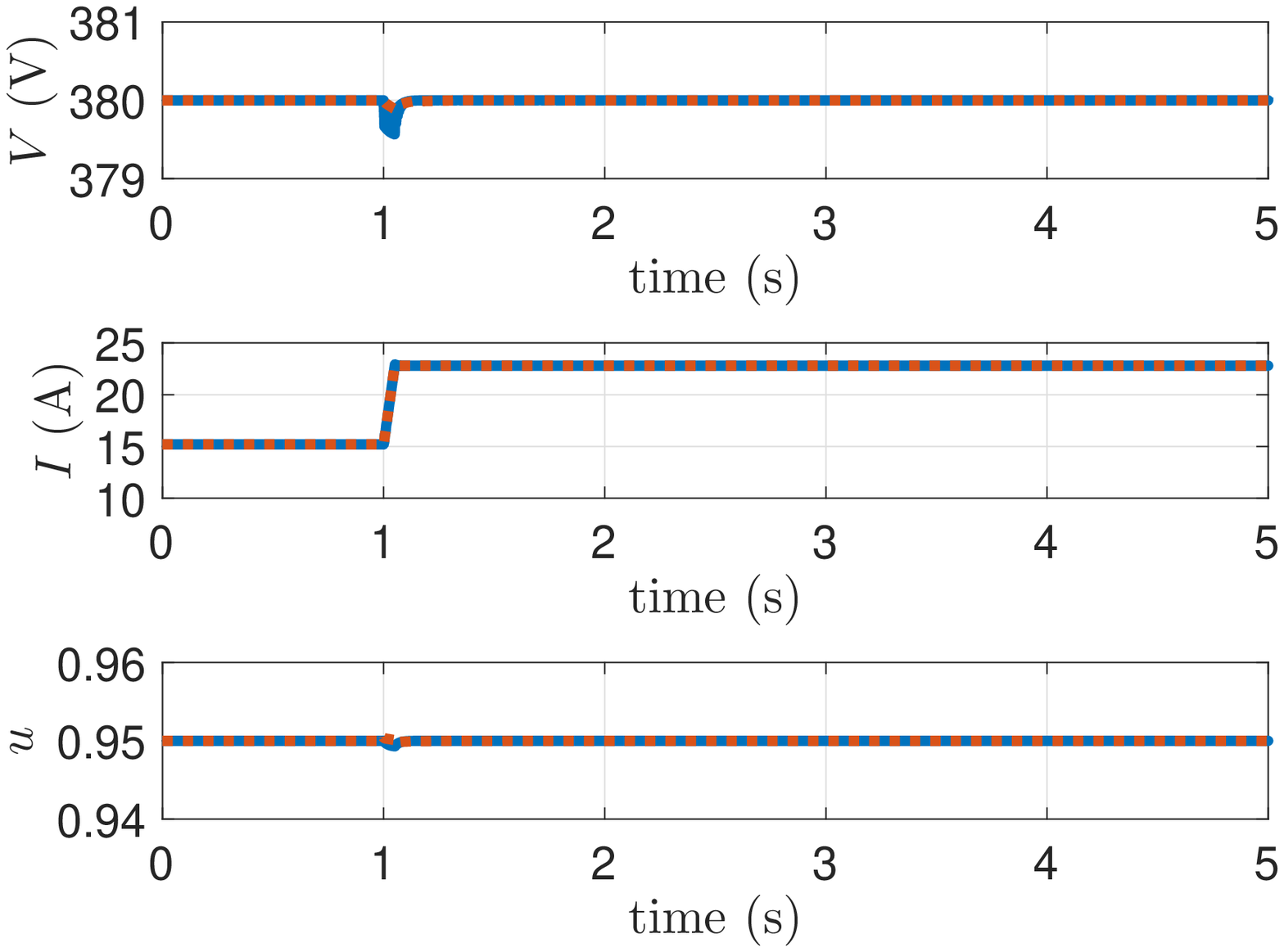}
\caption{(Input shaping for the buck converter) From the top: time evolution of the voltage, current and duty cycle considering a load variation $\Delta G$ at the time instant $t=$ 1 s. Input shaping for buck converter is plotted in blue color, while {\em Parallel Damping PBC} approach proposed in \cite{1323174} is plotted in red-dashed. (Parameters: $L =$ \SI{1}{\milli\henry}, $C =$ \SI{1}{\milli\farad}, $V_s =$ \SI{400}{\volt}, $G =$ \SI{0.04}{\siemens}, $\Delta G =$ \SI{0.02}{\siemens}, $V^\star =$ \SI{380}{\volt}, $k_d =$ \num{16e5}, $k_i =$ \num{8e7}, $\overline u =V^\star/V_s$ and \red{gamma in \cite[Equation (19)]{1323174} is set to 0.97}). } \label{fig:ex_buck_input_shaping}
\end{figure}
%%%%%%%%%%%%%%%%%%%%%%%%
%%%%%%%%%%%%%%%%%%%%%%%% FIGURE boost_output_shaping
\begin{figure}
\includegraphics[width=\columnwidth]{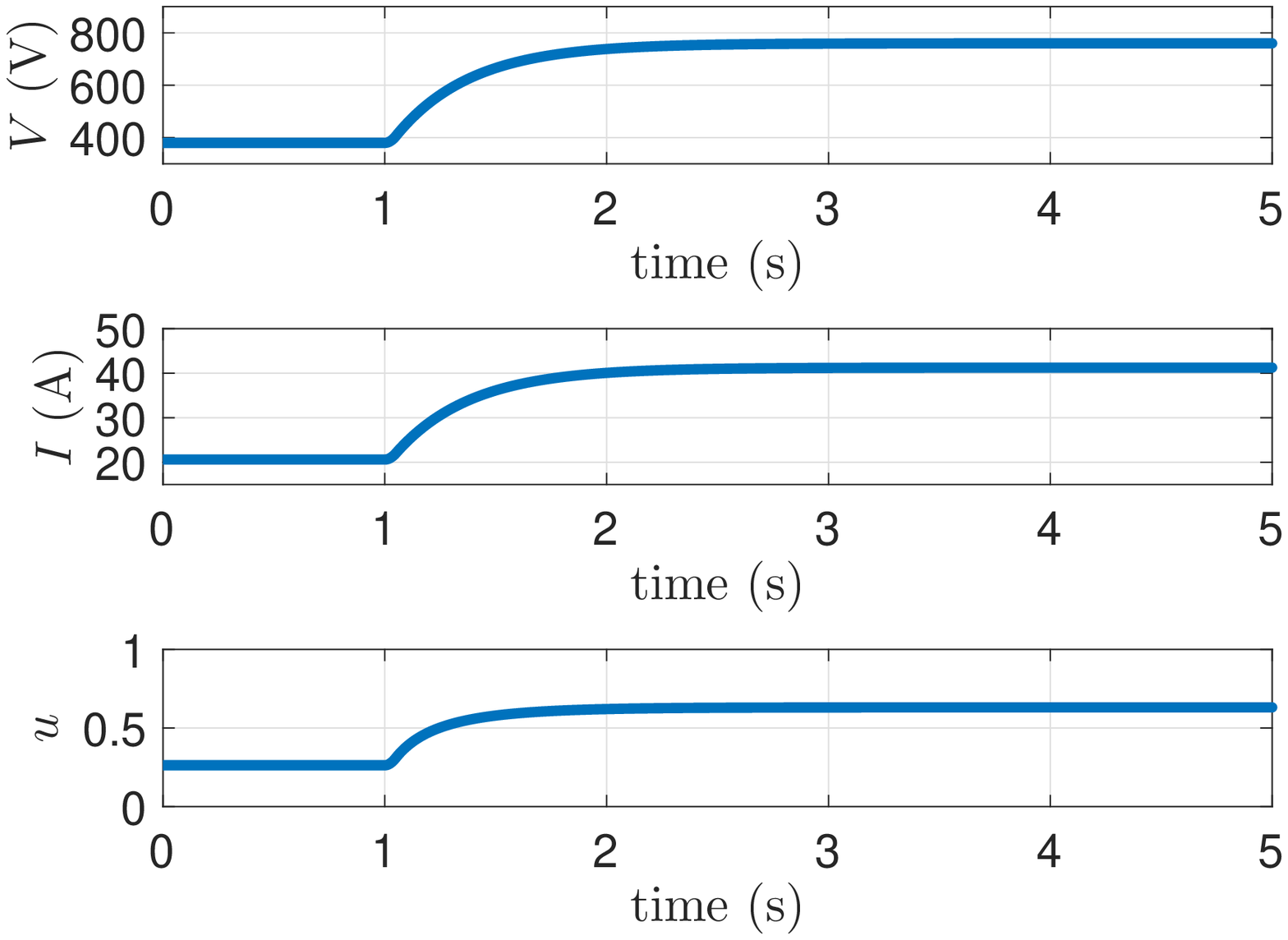}
\caption{(Output shaping for the boost converter) From the top: time evolution of the voltage, current and duty cycle considering a load variation at the time instant $t=$ 1 s (Parameters: $L =$ \SI{1.12}{\milli\henry}, $C =$ \SI{6.8}{\milli\farad}, $V_s =$ \SI{280}{\volt}, $G =$ \SI{0.04}{\siemens}, $\Delta G =$ \SI{-0.02}{\siemens}, $V^\star =$ \SI{380}{\volt}, $k_d =$ \num{5e2}, $k_i =$ \num{1e6}).} \label{fig:ex_boost_output_shaping}
\end{figure}
%%%%%%%%%%%%%%%%%%%%%%%% FIGURE buck_input_shaping
\begin{figure}
\includegraphics[width=\columnwidth]{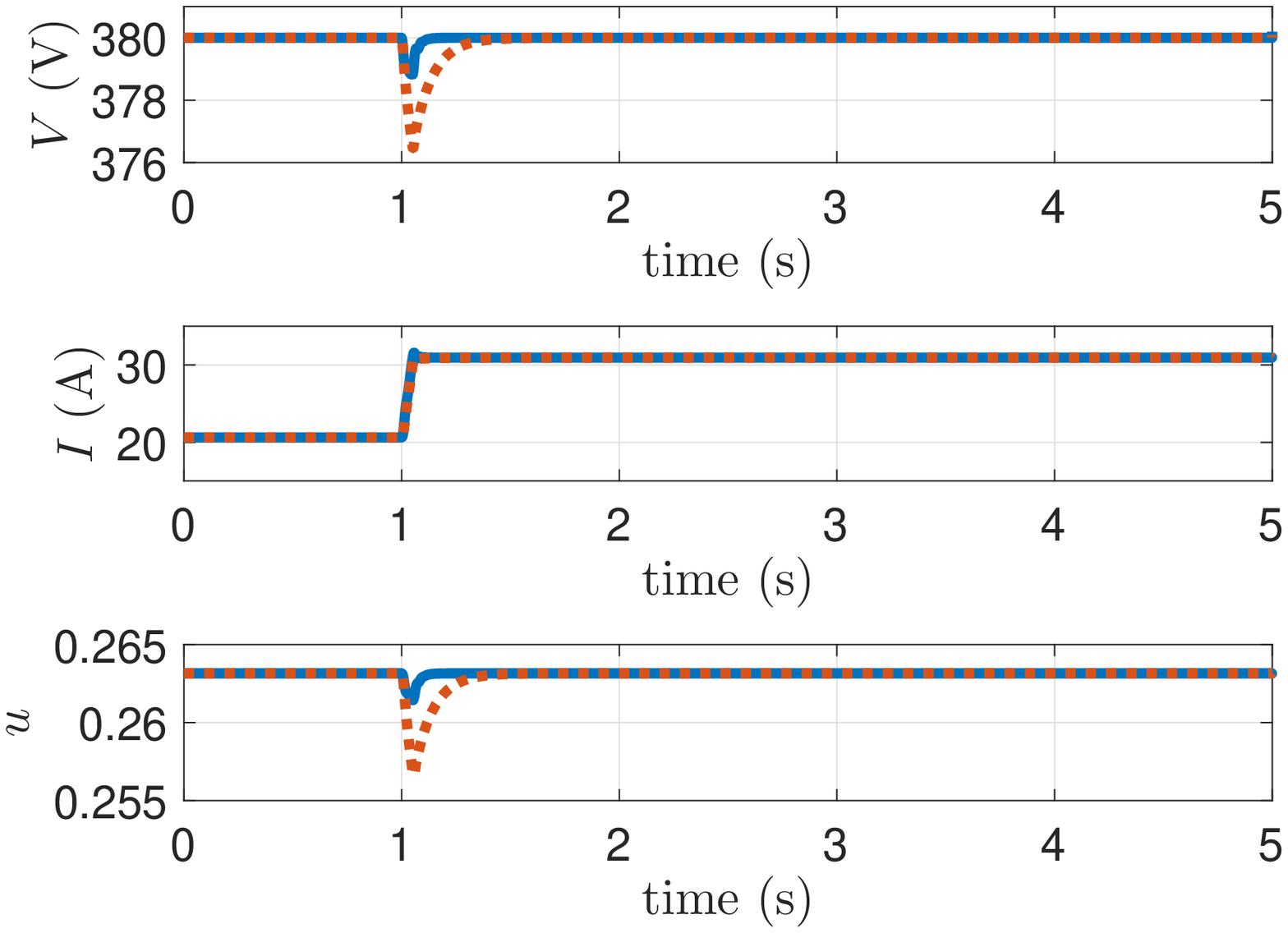}
\caption{(Input shaping for the boost converter) From the top: time evolution of the voltage, current and duty cycle considering a load variation at the time instant $t=$ 1 s. Input shaping for boost converter is plotted in blue color, while {\em Parallel Damping PBC} approach proposed in \cite{1323174} is plotted in red-dashed. (Parameters: $L =$ \SI{1.12}{\milli\henry}, $C =$ \SI{6.8}{\milli\farad}, $V_s =$ \SI{280}{\volt}, $G =$ \SI{0.04}{\siemens}, $\Delta G =$ \SI{0.02}{\siemens}, $V^\star =$ \SI{380}{\volt}, $k_d =$ \num{1e6}, $k_i =$ \num{4e7}, $\overline u =1-V_s/V^\star$ and \red{gamma in \cite[Equation (23)]{1323174} is set to 0.1}.)} \label{fig:ex_boost_input_shaping}
\end{figure}
%%%%%%%%%%%%%%%%%%%%%%%%
                %
                %
                %
                %
                %========================================CONTROLLERS
                \subsection{DC Networks}
                \label{sec:networks}

                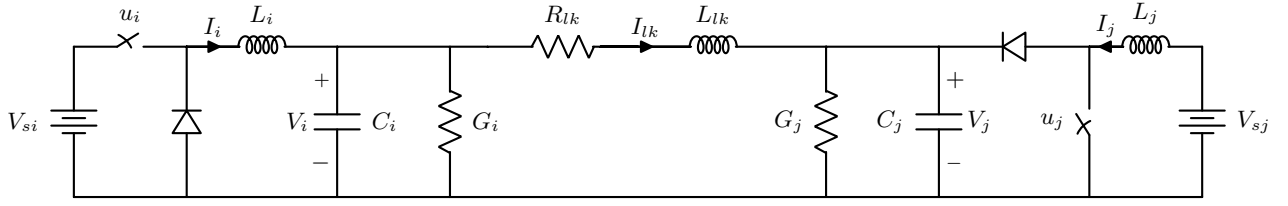
\begin{figure*}
                \begin{center}
                \ctikzset{bipoles/length=1cm}
                \begin{circuitikz}[scale=1,transform shape]
                \draw
                (0,1) node [] {} to [cspst=$u_{i}$] (1.5,1)
                (1.5,-1) {to [diode, l=, v>={{}}] (1.5,1)}
                (0,-1) node [] {} -- (5.0,-1)
                (0,1) {to [battery, l_=$V_{si}$] (0,-1)}
                (1.5,1) to [L, l=$L_{i}$, i>^={{$I_{i}$}}] (3.5,1)
                (3.5,1) {to [C, l=$C_{i}$, v>={{$V_{i}$}}] (3.5,-1)}
                (5.0,-1) {to [R, l_=$G_{i}$] (5.0,1)}
                (3.5,1) -- (5.5,1)
                (5.5,1) to [R, l=$R_{lk}$] (7.5,1)
                (7.0,1) to [L, l=$L_{lk}$,  i>^=$I_{lk}$] (10,1)
                (10,-1) {to [R, l=$G_{j}$] (10,1)}
                (11.5,-1) {to [C, l=$C_{j}$, v<={{$V_{j}$}}] (11.5,1)}
                (10,1) -- (11.5,1)
                (13.5,1) to [diode] (11.5,1)
                (13.5,-1) {to [cspst=$u_{j}$] (13.5,1)}
                (15,1) {to [battery, l=$V_{sj}$] (15,-1)}
                (13.5,1) to [L, l=$L_{j}$,  i<=$I_{j}$] (15,1)
                (5,-1) -- (15,-1);
            \end{circuitikz}
        \end{center}
        \caption{The considered electrical diagram of a (Kron reduced) DC network representing node $i \in \mathcal{V}_\alpha$ and node $j \in \mathcal{V}_\beta$ interconnected by the line $k \in \mathcal{E}$.}
        \label{fig:microgrid}
    \end{figure*}

    In this subsection we consider a typical DC microgrid of which a schematic electrical diagram is provided in Figure~\ref{fig:microgrid}, including a buck and boost DC-DC power converter interconnected through resistive-inductive power lines. 
    In the following we adopt the subscripts $\alpha$ or $\beta$ in order to refer to the buck or boost type converter, respectively.
    The network consists of $n_\alpha$ buck converters and $n_\beta$ boost converters, such that the total number of converters is $n_\alpha + n_\beta = n$.
    The overall network is represented by a connected and undirected graph $\mathcal{G} = (\mathcal{V}_\alpha \cup \mathcal{V}_\beta,\mathcal{E})$, where $\mathcal{V}_\alpha = \{1, \dots, n_\alpha \}$ is the set of the buck converters, $\mathcal{V}_\beta = \{n_\alpha+1, \dots, n \}$ is the set of the boost converters and $\mathcal{E}  = \{1,...,m\}$ is the set of the distribution lines interconnecting the $n$ converters.
    The network topology is  represented by its corresponding incidence matrix $\mathrm{D} \in \R^{n \times m}$. The ends of edge $k$ are arbitrarily labeled with a $+$ and a $-$, and the entries of  $\mathrm{D}$ are given by
    \begin{equation*}
        \label{eq:incidence}
        \mathrm{D}_{ik}=
        \begin{cases}
            +1 \quad &\text{if $i$ is the positive end of $k$}\\
            -1 \quad &\text{if $i$ is the negative end of $k$}\\
            0 \quad &\text{otherwise}.
        \end{cases}
    \end{equation*}

    According to \eqref{eq:buck}, the {average} dynamic equations of the  buck converter $i \in \mathcal{V_\alpha}$ become
    \begin{align}
        \begin{split}
            \label{eq:bucki}
            -L_{i}\dot{I}_{i} &=V_{i} - u_{i}V_{si}\\
            C_{i}\dot{V}_{i} &=I_{i} - G_{i}V_{i} -  \displaystyle{ \sum_{k \in \mathcal{E}_{i}}^{}I_{lk}},
        \end{split}
    \end{align}
    where $\mathcal{E}_{i} \subset \mathcal{E}$ is the set of the distribution lines incident to the  node $i$, and $I_{lk}$ denotes the current through the line $k\in \mathcal{E}_{i}$.
    On the other hand, according to \eqref{eq:boost}, the {average} dynamic equations of the boost converter $i \in \mathcal{V_\beta}$ become
    \begin{align}
        \begin{split}
            \label{eq:boosti}
            -L_{i}\dot{I}_{i} &=(1 - u_{i}) V_{i} - V_{si}\\
            C_{i}\dot{V}_{i} &=(1 - u_{i}) I_{i} - G_{i}V_{i} -  { \sum_{k \in \mathcal{E}_{i}}^{}I_{lk}}.
        \end{split}
    \end{align}
    %where $\mathcal{E}_{\beta i} \subset \mathcal{E}$ is the set of the distribution lines incident to the  boost converter $i$.
    %The current from DGu $i$ (including e.g. a boost converter) to DGu $j$ (including e.g. a buck converter) is denoted by $I_{ij}$, and its dynamic is given by
    %\begin{equation}
    %\label{eq:line}
    %L_{ij}{\dot I_{ij}} = (V_{\alpha_i} - V_{\beta_j}) - R_{ij} I_{ij}.
    %\end{equation}
The dynamic of the current $I_{lk}$ from node $i$ to
node $j \neq i$, $i, j \in \mathcal{V}_\alpha \cup
\mathcal{V}_\beta$, is given by
\begin{equation}
\label{eq:line} -L_{lk}{\dot I_{lk}} = -(V_{i} - V_{j}) + R_{lk}
I_{lk}.
\end{equation}
Let $V = [V_\alpha^\top, V_\beta^\top]^\top$, with $V_\alpha =
[V_{1}, \dots, V_{{n_\alpha}}]$ and $V_\beta = [V_{n_\alpha+1},
\dots, V_{{n}}]$. Analogously, let $I_\alpha =
[I_{1}, \dots, I_{{n_\alpha}}]$ and $I_\beta = [I_{n_\alpha+1},
\dots, I_{{n}}]$. To study the interconnected DC network, we write 
\eqref{eq:bucki}-\eqref{eq:line} compactly for all buses $i \in
\mathcal{V}_\alpha \cup \mathcal{V}_\beta$
%    \begin{align}
%        \label{eq:plant}
%        \begin{split}
%            -L_{\alpha}\dot{I}_{\alpha} & =V_\alpha - u_\alpha \circ V_{s\alpha}\\
%            -L_{\beta}\dot{I}_{\beta} & =(\mathds{1}_{n_\beta}-u_\beta) \circ V_\beta - V_{s\beta}\\
%            -L_l\dot{I}_l & =\mathrm{D}^T V + R_l I_l\\
%            C_{\alpha}\dot{V}_\alpha & =I_{\alpha} - G_{\alpha}V_\alpha + \mathrm{D}_\alpha I_l\\
%            C_{\beta}\dot{V}_\beta & =(\mathds{1}_{n_\beta}-u_\beta) \circ I_{\beta} - G_{\beta}V_\beta + \mathrm{D}_\beta I_l,
%        \end{split}
%    \end{align}
    	\begin{subequations}          \label{eq:plant}
    	\begin{align} 
              -L_{\alpha}\dot{I}_{\alpha} & =V_\alpha - u_\alpha \circ V_{s\alpha}\label{eq:planta}\\
  -L_{\beta}\dot{I}_{\beta} & =(\mathds{1}_{n_\beta}-u_\beta) \circ V_\beta - V_{s\beta}\label{eq:plantb}\\
  -L_l\dot{I}_l & =\mathrm{D}^T V + R_l I_l\label{eq:plantc}\\
  C_{\alpha}\dot{V}_\alpha & =I_{\alpha} - G_{\alpha}V_\alpha + \mathrm{D}_\alpha I_l\label{eq:plantd}\\
  C_{\beta}\dot{V}_\beta & =(\mathds{1}_{n_\beta}-u_\beta) \circ I_{\beta} - G_{\beta}V_\beta + \mathrm{D}_\beta I_l,\label{eq:plante}
    	\end{align}  
    \end{subequations}
where $I_\alpha, V_\alpha, V_{s\alpha}, u_\alpha \in \R^{n_\alpha}$,
$I_\beta, V_\beta, V_{s\beta}, u_\beta \in \R^{n_\beta}$, $I_l \in
\R^m$. Moreover, $L_\alpha, L_\beta, L_l, C_\alpha, C_\beta, R_l,
G_\alpha, G_\beta$, are positive definite diagonal matrices of
appropriate dimensions, e.g. $L_\alpha = \mathrm{diag}(L_1, \dots,
L_{n_\alpha})$, and $\mathds{1}_{n_\beta} \in \R^{n_\beta}$ denotes the
vector consisting of all ones. The matrices
$\mathrm{D}_\alpha \in \R^{n_\alpha \times m}$ and
$\mathrm{D}_\beta \in \R^{n_\beta \times m}$ are obtained by
collecting from $\mathrm{D}$ the rows indexed by
$\mathcal{V}_\alpha$ and $\mathcal{V}_\beta$, respectively.
Let $I = [I_\alpha^\top, I_\beta^\top, I_l^\top]^\top$, $u=[u_\alpha^\top, u_\beta^\top]^\top$, $V_s=[V_{s\alpha}^\top, V_{s\beta}^\top]^\top$, $L = \diag(L_\alpha, L_\beta, L_l)$ and $C = \diag(C_\alpha, C_\beta)$. We notice that system \eqref{eq:plant} can be expressed in the BM
formulation \eqref{BM_switched} with
\begin{equation}
\label{eq:B_network}
 B(u) =
\begin{bmatrix}
\diag(u_\alpha) & \vec{0}^{n_\alpha \times n_\beta}\\
\vec{0}^{n_\beta \times n_\alpha} & \mathds{I}_{n_\beta}\\
\vec{0}^{m \times n_\alpha} & \vec{0}^{m \times n_\beta}
  \end{bmatrix},
\end{equation}
and
\begin{align}
\label{eq:PT}
\begin{split} 
 P(u,I,V)=&~I^\top\Gamma(u) V+\frac{1}{2}I_l^\top R_l I_l\\
 &-\frac{1}{2}V_\alpha^\top G_\alpha V_\alpha - \frac{1}{2}V_\beta^\top G_\beta V_\beta,
\end{split}
\end{align}
where $\Gamma \in \R^{(n+m)\times n}$ is given by
\begin{equation}
\label{eq:Psi}
 \Gamma(u) =
\begin{bmatrix}
 \mathds{I}_{n_\alpha} & \vec{0}^{n_\alpha \times n_\beta}\\
 \vec{0}^{n_\beta \times n_\alpha} & \mathds{I}_{n_\beta}-\diag(u_\beta)\\
 \mathrm{D}_\alpha^T & \mathrm{D}_\beta^T
  \end{bmatrix},
\end{equation}
$\mathds{I}$ being the identity matrix. By using now the storage function in \eqref{com_storage}, the following passivity property for the considered DC network  \eqref{eq:plant} is established.
\begin{lemma}[Passivity property of DC Networks]\label{lem:DC_network_passivity}
Let Assumptions \ref{ass:LC} and \ref{ass:diss_potential} hold.	System \eqref{eq:plant} is passive with respect to the storage function \eqref{com_storage} and the port-variables $\dot{u}$ and 
	\begin{equation}\label{y_dc}
y_{\mathrm{DC}}=\begin{bmatrix}
\dot{I}_{\alpha}\circ V_{s\alpha}\\ \dot{I}_{\beta}\circ V_{\beta}-\dot{V}_{\beta}\circ I_{\beta}
\end{bmatrix}.
	\end{equation}
%	$y_{\mathrm{DC}}=\begin{bmatrix}
%\dot{I}_{\alpha}^\top V_s\\ \dot{I}_{\beta}V_{\beta}-\dot{V}_{\beta}I_{\beta}
%	\end{bmatrix}$.	
	%The system of equations \eqref{eq:buck} describing buck converter dynamics are passive with port-variables $\dot{u}$ and $\dot{I}V_s$.% or more conveniently $E\dot{u}$ and $\dot{I}$.
\end{lemma}
By virtue of the above passivity property, we can now use the input shaping  methodology to design a decentralized control scheme for regulating the voltage of \eqref{eq:plant}.

\begin{proposition}[Input shaping for DC Networks]\label{prop:input_shaping_dc_network}
Let Assumptions \ref{ass:LC}--\ref{ass:avail_information},  \ref{ass:feasibility_sRLC} and \ref{ass:input Matrix ass} hold. Consider system \eqref{eq:plant}  with the dynamic controller
	\begin{equation}\label{cont_dc_network_input}
	\dot{u}=-K_{d}^{-1}\left(K_{i}\left(u-\bar{u} \right)+ y_{\mathrm{DC}}\right),
%	\dot{u}_{\beta}&=&-\dfrac{1}{k_d}\left(k_i\left(u-\bar{u} \right)+ \left(\dot{I}V-\dot{V}I\right)\right)
	\end{equation}
	 where $K_d$ and $K_i$ are positive definite diagonal matrices of order $n_{\alpha}+n_{\beta}$, and $y_{\mathrm{DC}}$ is given by \eqref{y_dc}. Then, the solution $(I,V,u)$ to the closed-loop system asymptotically converges to the desired  steady-state  $\left(\overline{I},V^\star,\overline{u}\right)$.
\end{proposition}
\begin{proof}
	Consider the storage function \eqref{com_storage}. We use the integrated input port-variable to shape the desired closed-loop storage function, i.e.,
	\begin{equation}
	S_d= S+\dfrac{1}{2}\left(u-\bar{u}\right)^\top K_i \left(u-\bar{u}\right).
	\end{equation}
	Then, the first time derivative of $ S_d$ along the trajectories of system \eqref{eq:plant} controlled by \eqref{cont_dc_network_input} satisfies
	\begin{subequations}  \label{Sdot_IS_sRLC_Network}
		\begin{align} 
		\dot{S}_d&=-\dot{I}_l^\top R_l\dot{I}_l-\dot{V}^\top G\dot{V}+ \dot{u}^\top y_{DC}+\dot{u}^\top K_i \left(u-\overline{u}\right)\label{Sdot_IS_sRLC_Networka}\\
	%	&=-\dot{I}_l^\top R_l\dot{I}_l-\dot{V}^\top G\dot{V}+\dot{u}^\top  \left(y_{DC}+K_i\left(u-\overline{u}\right)\right)\label{Sdot_IS_sRLC_Networkb}\\
		&= -\dot{I}_l^\top R_l\dot{I}_l-\dot{V}^\top G\dot{V}-\dot{u}^\top K_d\dot{u},               \label{Sdot_IS_sRLC_Networkc}
%		&\leq \mu \dot{u}.\nonumber
		\end{align}  
	\end{subequations}
	where we use Lemma \ref{lem:DC_network_passivity} and the controller \eqref{cont_dc_network_input}. Then, from \eqref{Sdot_IS_sRLC_Networkc} there exists a forward invariant set $\Pi$ and by LaSalle's invariance principle the solutions that start in $\Pi$  converge to the largest invariant set contained in 
	\begin{equation}\label{set:forward_inv_set_thm_network}
	\Pi \cap \left\{\left(I,V,\dot{I},\dot{V},u\right):\dot{I}_l=0,\dot{V}=0,\dot{u}=0 \right\}.
	\end{equation}
On this invariant set, by differentiating \eqref{eq:plantd} and \eqref{eq:plante}  we get $\dot{I}=0$. Moreover, from \eqref{cont_dc_network_input} it follows that $u=\overline{u}$, which further implies $V=V^\star$ and $I=\overline{I}$. 
\end{proof}

The proposed decentralized control scheme is now assessed in simulation\footnote{For the readers interested also in experimental results obtained by implementing the input shaping control methodology in a real DC microgrid comprising boost converters, we refer to \cite{Cucuzzella_CDC19}.}, considering a DC network comprising four power converters (i.e., two buck and two boost converters) interconnected as shown in Figure~\ref{fig:microgrid_example}.
The parameters of the converters and lines are reported in Table \ref{tab:parameters1} and~\cite[Table III]{cucuzzella2017robust}, respectively.
The controller gains for the buck converters are $k_{d\alpha} =$ \num{4e5} and $k_{i\alpha} =$ \num{4e7}, while for the boost converters are $k_{d\beta} =$ \num{1e6} and $k_{i\beta} =$ \num{4e7}. The most significant electrical signals of the simulation results are shown in Figure \ref{fig:sim_network}. In order to verify the robustness property of the control scheme with respect to the load uncertainty, the value of the load is changed from $G$ to $G+\Delta G$ at the time instant $t=$ \SI{1}{\second} (see Table \ref{tab:parameters1}). One can appreciate that the input shaping methodology is robust with respect to load uncertainty (see Remark \ref{rem:robustness}).

%\end{remark}
%%%%%%%%%%%%%%%%%%%%%%%% FIGURE network   
\begin{table}[t]
\caption{Network Parameters}
\centering
{\begin{tabular}{lc | cccc}			
Node								&	&1		&2	&3	&4\\			
\hline
$L_{i}$	&(\si{\milli\henry})	&\num{1.8}		&\num{1.12}		&\num{3.0}		&\num{1.12}\\
$C_{i}$	&(\si{\milli\farad})	&\num{2.2}		&\num{6.8}		&\num{2.5}		&\num{6.8}\\
$V_{si}$ &(\si{\volt})		&\num{400.0}		&\num{280.0}		&\num{450.0}		&\num{320.0}\\
$V_i^{\star}$ &(\si{\volt})		&\num{380.0}		&\num{380.0}		&\num{380.0}		&\num{380.0}\\
$G$ &(\si{\siemens})	&\num{0.08}		&\num{0.04}		&\num{0.05}		&\num{0.07}\\
$\Delta G$ &(\si{\siemens})	&\num{0.01}		&\num{0.03}		&\num{-0.03}		&\num{0.01}
\end{tabular}}
\label{tab:parameters1}
\end{table}
%\vspace{0.2cm}
%%\begin{table}[t]
%\caption{Line Parameters}
%\centering
%{\begin{tabular}{lc | cccc}			
%Line								&	&1&2	&3		&4\\					
%\hline
%% & & & & &\\
%$R_{lk}$	&(\si{\milli\ohm})	&\num{70}		&\num{50}		&\num{80}		&\num{60}\\
%$L_{lk}$	&(\si{\micro\henry})	&\num{2.1}		&\num{2.3}		&\num{2.0}		&\num{1.8}\\
%\end{tabular}}
%\label{tab:parameters2}
%\end{table}
 
    \begin{figure}
\centering
\begin{tikzpicture}[scale=0.8,transform shape,->,>=stealth',shorten >=1pt,auto,node distance=3cm,
                    semithick]
  \tikzstyle{every state}=[circle,thick,fill=white,draw=black,text=black]

  \node[state] (A)                    {\num{1}};
  \node[state]         (B) [above right of=A] {\num{2}};
  \node[state]         (D) [below right of=A] {\num{4}};
  \node[state]         (C) [below right of=B] {\num{3}};

  \path (A) edge   [below] node {\hspace{7mm}$I_{l1}$} (B)
  		edge 	     node {$I_{l4}$} (D)
           (B) edge      [below]        node {\hspace{-7mm}$I_{l2}$} (C)
           (C) edge         [above left]     node {$I_{l3}$} (D);

\end{tikzpicture}
\caption{Scheme of the considered network with 4 power converters: Nodes 1 and 3 have buck converters, Nodes 2 and 4 have boost converters. }\label{fig:microgrid_example}
\end{figure}
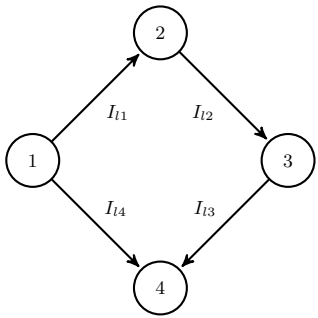

%%%%%%%%%%%%%%%%%%%%%%%% FIGURE network sim
\begin{figure}
\includegraphics[width=\columnwidth]{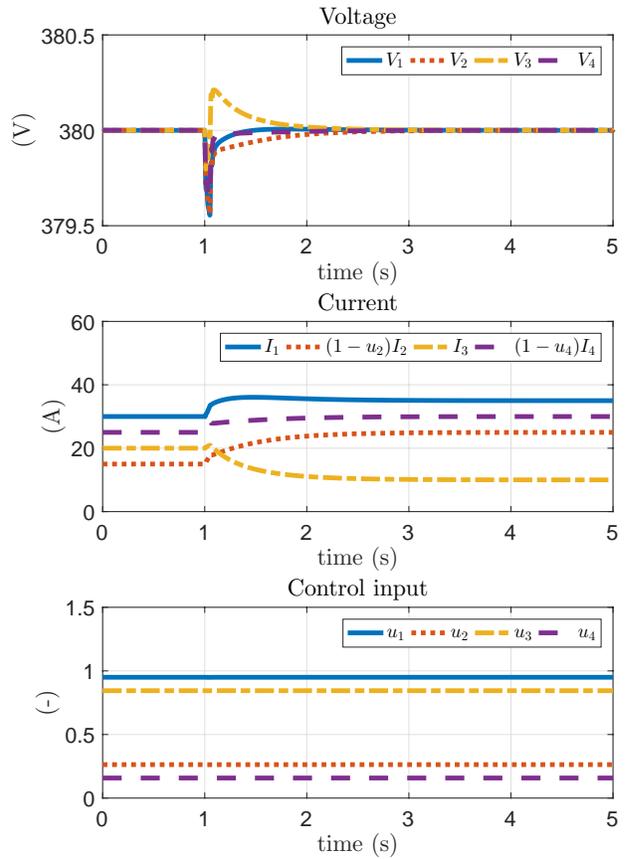}
\caption{(Input shaping for the DC network) From the top: time evolution of the voltage of each node, current generated by each converter and duty cycle of each converter, considering a load variation at the time instant $t=$ 1 s.} \label{fig:sim_network}
\end{figure}
%%%%%%%%%%%%%%%%%%%%%%%%

    %
    %
    %
    %
    %========================================CONCLUSIONS
    \begin{table*}[t]
    	\centering
    	\caption{supply-rates of RLC and s--RLC circuits}
    	{\begin{tabular}{llllll}
    			\toprule
    			framework& &\hspace{3.5cm}supply-rate&&\\
    			\midrule
    			&RLC& \hspace{3.5cm}s--RLC& buck & boost&buck-boost\\
    			\midrule
    			Port--Hamiltonian                     & $I^\top Bu_s$ & $I^\top B(u)V_s$ = $I^\top B_0V_s+ u I^\top \left(B_1-B_0\right)V_s$&$uI V_s$&$IV_s$&$uIV_s$\\
    			Brayton--Moser                  & $\dot{I}^\top Bu_s$ & \hspace{2cm}~~\;\;-\;&$u\dot{I}V_s$&-&-\\
    			Proposed  &  $\dot{I}^\top B\dot{u}_s$ & $\dot{u}(\dot{V}^\top(\Gamma_1-\Gamma_0)^\top
    			I-\dot{I}^\top(\Gamma_1-\Gamma_0)V-\dot{I}^\top
    			(B_0-B_1)V_s)$&$\dot{u}\dot{I}V_s$&$\dot u\left(\dot{I}V-\dot{V}I\right)$& $\dot{u}\left(\dot{I}V-\dot{V}I+V_{s}\dot{I}\right)$\\
    			%	Differential-passivity &  $\delta{I}^\top B\delta{u}_s$ &\\
    			\midrule
    			\bottomrule
    	\end{tabular}}
    	\label{tab:supplyrates}
    \end{table*}
    \section{Conclusions and Future Works}
    \label{sec:conclusions}
In this paper, we have presented new passivity properties for a class of RLC and s--RLC circuits that are  modeled using the Brayton-Moser formulation. 
%Moreover, Assumptions \ref{ass:LC} - \ref{ass:feasibility_sRLC} identifies these class of the systems, which can be easily verified. 
We use these new passivity properties to propose two new control methodologies: \emph{output shaping} and \emph{input shaping}. The  key observations are:
\begin{itemize}
%\item [(i)] The obtained passive map in Propositions \ref{prop: passivity_Non_switched electrical circuits} and \ref{prop: passivity_switched electrical circuits}, consider only the dynamics evolving in the tangent space of the system. In this way, the storage function \eqref{com_storage} can be elucidated as the energy in the tangent space.
%\item [(ii)] In Table \ref{tab:supplyrates} we have summerized the passivity properties derived in the port-Hamiltonian, Brayton-Moser and proposed frameworks, respectively. As shown in Remark \ref{energy_s-rlc}, in order to infer new passive maps for s--RLC circuits using the total energy as storage function, matrices $B_0$ and $B_1$ need to be zero and non-zero, respectively. 
\item [(i)]% In PID-PBC, one derives new (useful) open-loop passivity properties with (possibly) integrable outputs \cite{bibid}. The Output Shaping methodology presented in Theorem \ref{prop:Non_Switched_Method_1} and \ref{prop:Switched_Method_1} is very similar to PID-PBC technique.  
The output shaping methodology exploits the integrability property of the output port-variable. The input shaping technique instead exploits the integrability property of the input port-variable.
\item [(ii)] The controllers based on the input shaping methodology show robustness properties with respect to load uncertainty.
\end{itemize}
{Possible future directions include to incorporate nonlinear loads (e.g. constant power loads \cite{Cucuzzella_CDC19,Cucuzzella_AUT_ZIP}), develop distributed control schemes (e.g. for achieving load sharing \cite{8814756}) and extend such a new passivity concept to a wider class of nonlinear systems \cite{kosaraju2019krasovskiis,Yu2019}.}

%Possible future directions include to 
%investigate the relationship between the proposed methodologies and similar passivity properties (e.g.  differential passivity, incremental passivity and equilibrium independent passivity \cite{5531026}), 
%incorporate nonlinear loads (e.g. constant power loads \cite{Cucuzzella_CDC19,Cucuzzella_AUT_ZIP}) and develop distributed control schemes (e.g. for achieving load sharing \cite{8814756}). 
%\begin{itemize}
%	\item[(i)] Investigate advantages/disadvantages with respect to similar passivity properties, such as, differential passivity, incremental passivity and equilibrium independent passivity \cite{5531026}.
%%	\item[(ii)] Include constant power loads.
%%	\item[(iii)] The line dynamics considered in Section \ref{sec:networks} are Resistive and Inductive. But, a more accurate model for line-dynamics are given by Telegrapher's equations. This calls for an extension of the proposed  techniques to interconnected finite-infinite dimensional systems.
%\end{itemize}

%Furthermore, we have proposed two control strategies that make use of these open loop passive maps: i) Output Shaping, and ii) input shaping. 
    %
    %
    %
    %
    %========================================APPENDIX
    \appendix
In this Appendix we use the input shaping methodology to design voltage controllers for the buck-boost and C\'uk converters, respectively.
    The proofs of the following Corollaries are analogous to those of Corollaries \ref{prop:input_shaping_buck} and \ref{prop:input_shaping_boost} presented in Section~\ref{sec:controllers}.
    \subsection{Buck-boost Converter}
    \label{appendixA}
    The {average} governing dynamic equations of the buck-boost converter are the following:
    \begin{align}
        \begin{split}
            \label{eq:buckboost}
            -L\dot{I} &=(1 - u) V - uV_s\\
            C\dot{V} &=(1 - u) I - G V.
        \end{split}
    \end{align}
Equivalently, system \eqref{eq:buckboost} can be obtained from
\eqref{BM_switched2} with $\Gamma_0=1$, $\Gamma_1=0$, $B_0=0$, $B_1=1$ and $R=0$. 
    By using Proposition~\ref{prop: passivity_switched electrical circuits}, the following passivity property is established.
\begin{lemma}[Passivity property of the buck-boost converter]
Let Assumptions \ref{ass:LC} and \ref{ass:diss_potential} hold. System \eqref{eq:buckboost} is passive with respect to the storage function \eqref{com_storage} and the port-variables $\dot{u}$ and $y=\dot{I}V-\dot{V}I+V_{s}\dot{I}$.
\end{lemma}
By virtue of the above passivity property, we can now use the input shaping  methodology to design a voltage controller. 
\begin{corollary}[Input shaping for the buck-boost converter]\label{prop:input_shaping_buckboost}
	Let Assumptions \ref{ass:LC}--\ref{ass:avail_information},  \ref{ass:feasibility_sRLC} and \ref{ass:input Matrix ass} hold. Consider system \eqref{eq:buckboost} with the dynamic controller
\begin{equation}\label{cont_buckboost_input}
\dot{u}=-\dfrac{1}{k_d}\left(k_i\left(u-\overline{u} \right)+ \left(\dot{I}V-\dot{V}I+V_{s}\dot{I}\right)\right),
\end{equation}
	with $k_d>0$ and $k_i>0$. Then, the solution $(I,V,u)$ to the closed-loop system asymptotically converges to the desired steady-state  $\left(\overline{I},V^\star,\overline{u}\right)$.
\end{corollary}
%\begin{proof}
%		Proof of Corollaries \ref{prop:input_shaping_buckboost}, \ref{prop:input_shaping_cuk} follow along the same lines of input shaping controller for buck and boost systems presented in Corollaries \ref{prop:input_shaping_buck} and Corollary \ref{prop:input_shaping_boost} respectively. 
%%Proof of this and Corollary \ref{prop:input_shaping_buckboost}, \ref{prop:input_shaping_cuk} follows along the same lines of input shaping controller for buck and boost systems presented in Corollary \ref{prop:input_shaping_buck} and Corollary \ref{prop:input_shaping_boost} respectively.
%\end{proof}
\subsection{C\'uk Converter}
\label{appendixB}
The {average} governing dynamic equations of the C\'uk converter are the following:
\begin{subequations}   \label{eq:cuk}
\begin{align} 
            -L_1\dot{I}_1 &=(1 - u) V_1 - V_s \label{eq:cuka}\\
            -L_2\dot{I}_2 &=u V_1 + V_2\label{eq:cukb}\\
            C_1\dot{V}_1 &=(1 - u) I_1 + u I_2\label{eq:cukc}\\
            C_2\dot{V}_2 &=I_2 - G V_2. \label{eq:cukd}
    \end{align}
\end{subequations}
%     \begin{align}
%         \begin{split}
%             \label{eq:cuk}
%             -L_1\dot{I}_1 &=~(1 - u) V_1 - E\\
%             -L_2\dot{I}_2 &=~u V_1 + V_2\\
%             C_1\dot{V}_1 &=~(1 - u) I_1 + u I_2\\
%             C_2\dot{V}_2 &=~I_2 - G V_2,
%         \end{split}
%     \end{align}
Equivalently, system \eqref{eq:cuk} can be obtained from \eqref{BM_switched} with $\Gamma_0=\begin{bmatrix}
    1&0\\0&1
    \end{bmatrix}$, $\Gamma_1=\begin{bmatrix}
    0&0\\1&1
    \end{bmatrix}$, $B_0=B_1=[1\; 0]^\top $, $P_R(I)=0$ and $P_G(V)=\frac{1}{2}GV^2_2$.  By using Proposition~\ref{prop: passivity_switched electrical circuits}, the following passivity property is established.
    \begin{lemma}[Passivity property of the C\'uk converter]
 Let Assumptions \ref{ass:LC} and \ref{ass:diss_potential} hold.   	System \eqref{eq:cuk} is passive with respect to the storage function \eqref{com_storage} and the port-variables $\dot{u}$ and $y=\dot{V}_1\left(I_2-I_1\right)-V_1\left(\dot{I}_2-\dot{I}_1\right)$.
    \end{lemma}
By virtue of the above passivity property, we can now use the input shaping  methodology to design a voltage controller. 
\begin{corollary}[Input shaping for the C\'uk converter]\label{prop:input_shaping_cuk}
	Let Assumptions \ref{ass:LC}--\ref{ass:avail_information},  \ref{ass:feasibility_sRLC} and \ref{ass:input Matrix ass} hold. Consider system \eqref{eq:buckboost}  with the dynamic ocntroller
	\begin{equation}\label{cont_cuk_input}
	\dot{u}=-\dfrac{1}{k_d}\left(k_i\left(u-\overline{u} \right)+ \dot{V}_1\left(I_2-I_1\right)-V_1\left(\dot{I}_2-\dot{I}_1\right)\right),
	\end{equation}
	with $k_d>0$ and $k_i>0$. Then, the solution $(I,V,u)$ to the closed-loop system asymptotically converges to the desired steady-state  $\left(\overline{I},V^\star,\overline{u}\right)$.
\end{corollary}
%
%\begin{corollary}[Input Shaping of C\'uk converter]
%Define the following dynamics for input $u$
%\begin{eqnarray}\label{cont_cuk_input}
%\dot{u}:=-\dfrac{1}{k_d}\left(k_I\left(u-\overline{u} \right)+ y\right)
%\end{eqnarray}
%where $k_d>0$, $\overline{u}=\dfrac{V_2^\star}{V_2^\star-V_{s}}$. Then, the system of equations \eqref{eq:cuk} together with \eqref{cont_cuk_input} are asymptotically stable at the operating point $(I_1,I_2,V_1,V_2,u)=\left(\dfrac{-\overline{u}}{1-\overline{u}}GV_2^\star,GV^\star_2,\dfrac{V_{s}}{1-\overline{u}},V_2^\star,\overline{u}\right)$.
%\end{corollary}
%\begin{proof}
%One can also prove the asymptotic stability using $S_d$ in \eqref{storage_input_shaping_clp_b} as the closed-loop Lyapunov function. The time derivative of the storage function $S_d$ \eqref{storage_input_shaping_clp_b} is,
%                \begin{eqnarray}
%                \dfrac{d}{dt}S_d=-G\dot{V}^2_2-k_d\dot{u}^2
%                \end{eqnarray}
%Further the residual set of $\dot{S}_d$ is $\Pi=\left\{\left(I,V,u_s,\dot{I},\dot{V}\right)|\dot{S}_d=0\right\}$. Whenever a trajectory $(V(t),I(t),u(t))$ of system \eqref{eq:cuk} together with \eqref{cont_cuk_input} belongs to $\Pi \implies \dot{V}_2(t)= 0, ~ \dot{u}(t)= 0$, \eqref{eq:cukd}  $ \implies \dot{I}_2(t)= 0 $, \eqref{eq:cukb} $\implies~\dot{V}_1= 0$, \eqref{eq:cukc} $\implies ~ \dot{I}_1= 0$, $\implies y= 0$, \eqref{cont_cuk_input} $\implies u= \overline{u}$, $\implies ~V_2= V_2^\star$.
%\end{proof}
    %
    %
    %
    %
    %============================BIOGRAPHY
    \bibliographystyle{IEEEtran}
    \bibliography{references}
    \begin{IEEEbiography}[{\includegraphics[width=1in,height=1.25in,clip,keepaspectratio]{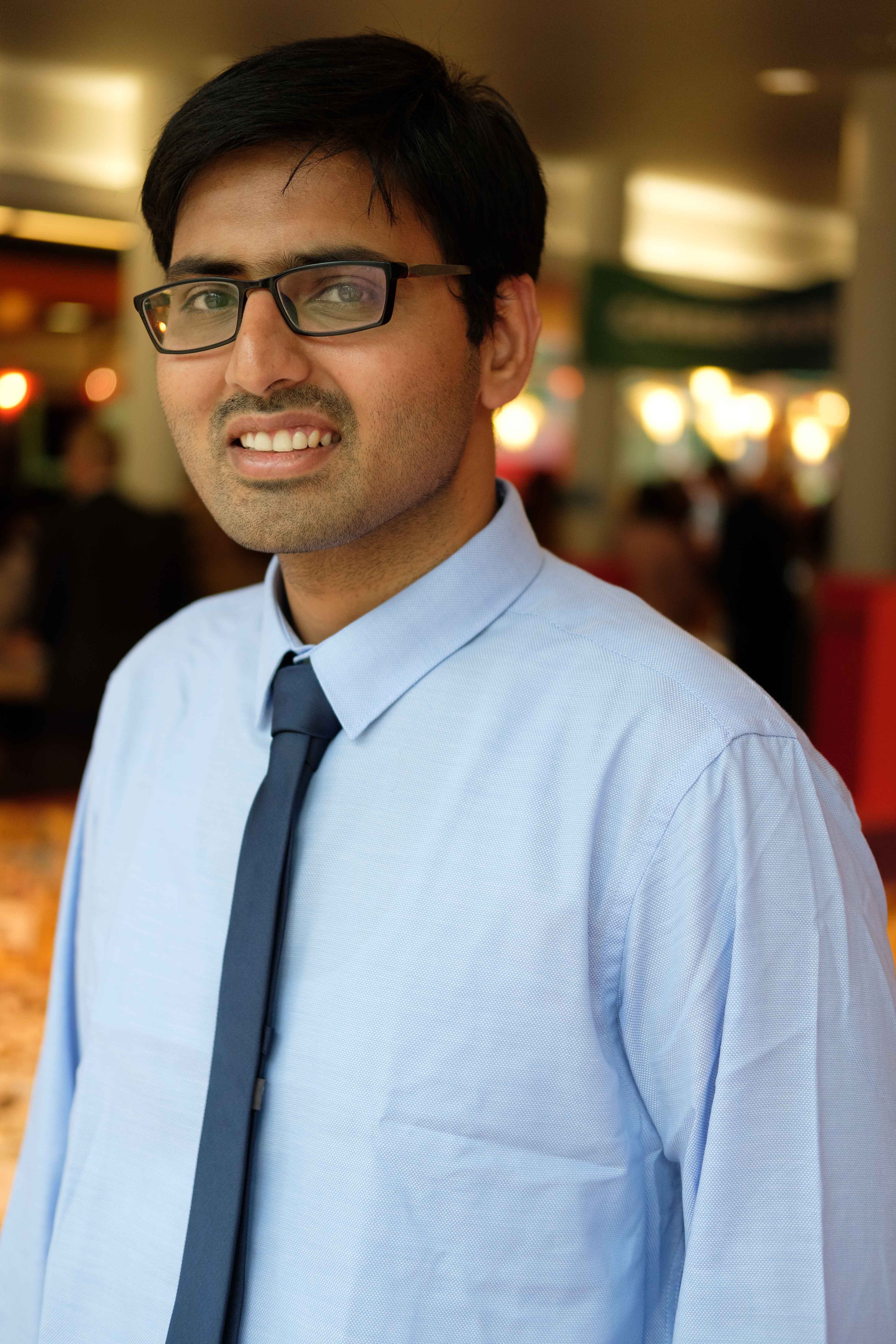}}]{Krishna Chaitanya Kosaraju} received the Bachelor degree in Electronics and Instrumentation from Birla Institute of Technology and Science - Pilani, 2010. He received his Master degree in Control and Instrumentation and the Ph.D. degree in Electrical Engineering from Indian Institute of Technology - Madras, in 2013 and 2018, respectively. Currently he is a Postdoc at the University of Notre Dame, under the
    	supervision of Professor Vijay Gupta. From April 2018 to December 2019, he was a Postdoc at the University of Groningen, under the
		supervision of Professor J.M.A. Scherpen. His research activities are mainly in the
		area of nonlinear control theory, passivity based control and optimization theory with application to power networks, Building systems and Reinforcement learning.
\end{IEEEbiography}
    \begin{IEEEbiography}[{\includegraphics[width=1in,height=1.5in,clip,keepaspectratio]{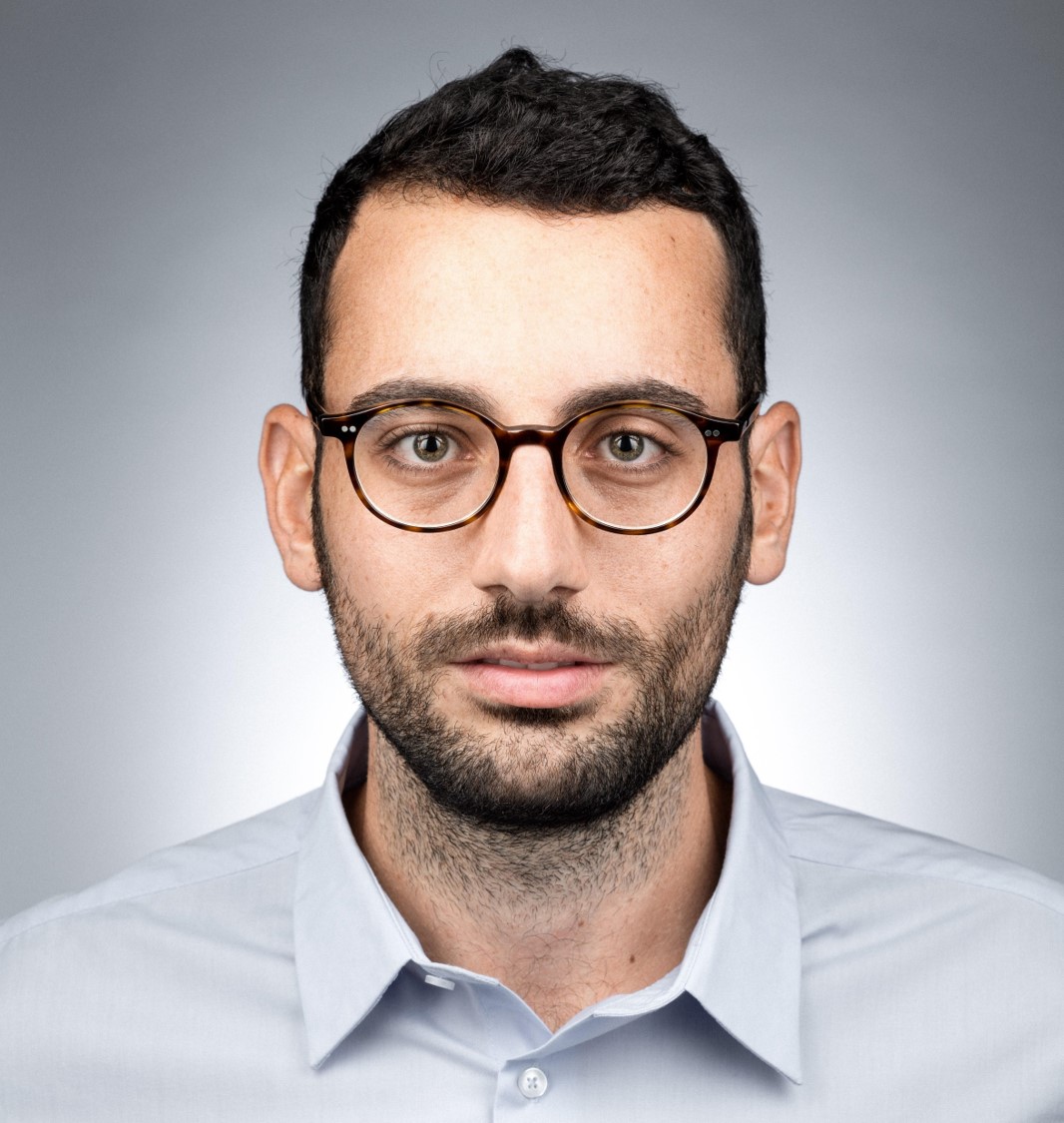}}]{Michele Cucuzzella} received the Bachelor degree (with highest honor) in
    		Industrial Engineering, the Master degree (with highest honor) in Electrical
    		Engineering and the Ph.D. degree (with highest honor) in Electronics, Computer
    		Science and Electrical Engineering from the University of Pavia, in 2012, 2014
    		and 2018, respectively. From April to June 2016, and from February to March
    		2017 he was with the Johann Bernoulli Institute for Mathematics and Computer
    		Science at the University of Groningen, under the supervision of Professor A. van
    		der Schaft. Currently he is a Postdoc at the University of Groningen, under the
    		supervision of Professor J. Scherpen. His research activities are mainly in the
    		area of nonlinear control  with
    		application to power networks. He is IEEE member since 2016 and, currently, he
    		is Associate Editor for the European Control Conference. He received the IEEE
    		Italy Section Award for the best Ph.D. thesis on new technological challenges in
    		energy and industry, and the SIDRA Award for the best Ph.D. thesis in the field of
    		systems and control engineering.
    \end{IEEEbiography}
    \begin{IEEEbiography}[{\includegraphics[width=1in,height=1.25in,clip,keepaspectratio]{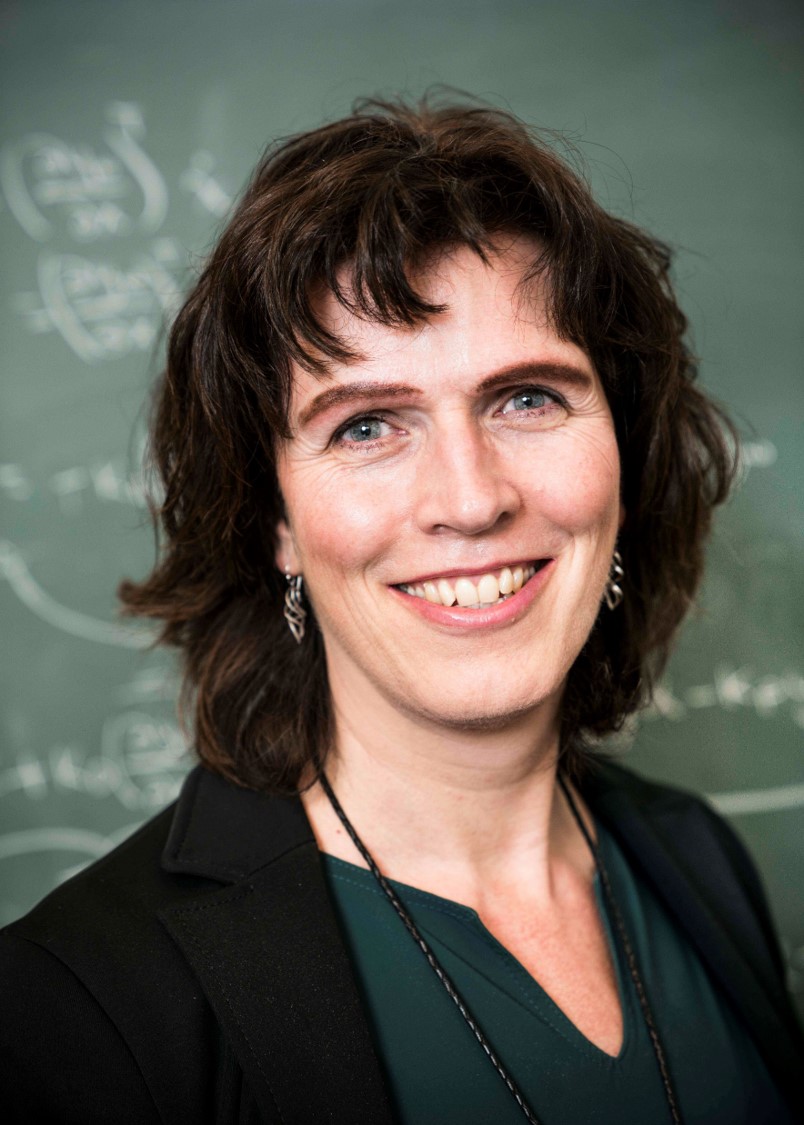}}]{Jacquelien M. A. Scherpen}  received the M.Sc. and Ph.D. degrees in applied mathematics from the University of Twente, Enschede, The Netherlands, in 1990 and 1994, respectively.
		She was with Delft University of Technology, The Netherlands, from 1994 to 2006. Since September 2006, she has been a Professor with the University of Groningen, at the Engineering and Technology institute Groningen (ENTEG) of the Faculty of Science and Engineering, Groningen, The Netherlands. She is currently scientific director of ENTEG (until 2019), and of the Groningen Engineering Center. She has held visiting research positions at The University of Tokyo, Japan, Universit\'e de Compiegne, SUPELEC, Gif-sur-Yvette, France, and the Old Dominion University, Norfolk, VA, USA. Her current research interests include nonlinear model reduction methods, nonlinear control methods, modeling and control of physical systems with applications to electrical circuits, electromechanical systems and mechanical systems, and distributed optimal control applications to smart grids.
		Prof. Scherpen has been an Associate Editor of the IEEE TRANSACTIONS ON AUTOMATIC CONTROL, the International Journal of Robust and Nonlinear Control (IJRNC), and the IMA Journal of Mathematical Control and Information. She is on the Editorial Board of the IJRNC.
\end{IEEEbiography}
    \begin{IEEEbiography}[{\includegraphics[width=1in,height=1.25in,clip,keepaspectratio]{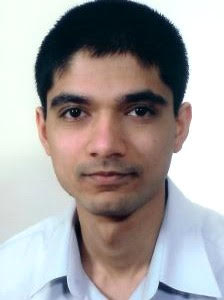}}]{Ramkrishna Pasumarthy} obtained his PhD in systems and control from the University of Twente, The Netherlands. He is currently with the Indian Institute of Technology Madras, India and is also associated with the Robert Bosch Center for data sciences and artificial intelligence at IIT Madras. His research interests are mainly in the area of modeling and control of complex physical systems, together with identification and control of (cloud) computing systems and data analytics for power, traffic, cloud and brain networks.
\end{IEEEbiography}
\balance

\end{document}